\documentclass[journal,twoside,web]{ieeecolor}
\usepackage{generic}
\usepackage{cite}
\usepackage{amsmath,amssymb,amsfonts}
\usepackage{graphicx}
\usepackage{hyperref}
\usepackage{textcomp}

\usepackage[ruled,vlined]{algorithm2e}
\usepackage{bm}
\usepackage{booktabs}
\usepackage{color}
\usepackage{multirow}
\usepackage{makecell}
\usepackage{subfigure}

\newtheorem{theorem}{Theorem}
\newtheorem{lemma}{Lemma}
\newtheorem{definition}{Definition}
\newtheorem{remark}{Remark}

\newtheorem{assumption}{Assumption}

\usepackage{array}
\renewcommand\arraystretch{0.75}

\def\BibTeX{{\rm B\kern-.05em{\sc i\kern-.025em b}\kern-.08em
    T\kern-.1667em\lower.7ex\hbox{E}\kern-.125emX}}
\begin{document}

\title{Resilient Quantized Consensus in \\Multi-Hop Relay Networks}
\author{Liwei Yuan, \IEEEmembership{Member, IEEE}, and Hideaki Ishii, \IEEEmembership{Fellow, IEEE} 
\thanks{This work was supported in part by the National Natural Science Foundation of China under Grant~62403188 and in part by JSPS under Grants-in-Aid for Scientific Research Grant No.~22H01508 and 24K00844.
	}
\thanks{L. Yuan is with the College of Electrical and Information Engineering, Hunan University, Changsha, 410082, China (e-mail: yuanliwei@hnu.edu.cn). }
\thanks{H. Ishii is with the Department of Information Physics and Computing, The University of Tokyo, Tokyo, 113-8656, Japan (e-mail: hideaki\_ishii@ipc.i.u-tokyo.ac.jp). }
}

\maketitle

\begin{abstract}
We study resilient quantized consensus in multi-agent systems, where some agents may malfunction. The network consists of agents taking integer-valued states, and the agents' communication is subject to asynchronous updates and time delays. We utilize the quantized weighted mean subsequence reduced algorithm where agents communicate with others through multi-hop relays. We prove necessary and sufficient conditions for our algorithm to achieve the objective under the malicious and Byzantine attack models. Our approach has tighter graph conditions compared to the one-hop algorithm and the flooding-based algorithms for binary consensus. Numerical examples verify the efficacy of our algorithm.
\end{abstract}


\section{Introduction}\label{secintro}
\IEEEPARstart{W}{ith}
the growth of concerns on cyber-security issues of multi-agent systems, distributed consensus in the presence of adversarial agents has been widely studied \cite{Lynch,teixeira2012attack,yuan2023event}. 
The faulty agents considered in existing works are of two main categories:
malicious and Byzantine agents. They can behave arbitrarily to prevent normal ones from achieving global objectives. Malicious agents must send the same false messages to their neighbors \cite{goldsmith2005wireless,leblanc2013resilient}, while Byzantine agents can send different messages to different neighbors \cite{Lynch,vaidya2012iterative}. 

We aim to solve resilient quantized consensus under the malicious and Byzantine attack models, where all agents take integer values. 
The objective is for nonfaulty, normal agents to reach consensus on a common value regardless of adversaries' misbehaviors.
Our approach exploits the mean subsequence reduced (MSR) algorithms \cite{leblanc2013resilient,vaidya2012iterative}, which have been widely used to solve resilient real-valued \textit{approximate} consensus \cite{usevitch2020determining,yuan2021resilient,sakavalas2020asynchronous,su2017reaching,yuan2022asynchronous}. 
For MSR algorithms to succeed, the graph property called robustness is critical \cite{leblanc2013resilient}.
Nevertheless, it requires relatively dense and complex networks \cite{usevitch2020determining}.
Recently, several works introduced multi-hop relays to MSR algorithms and consequently relaxed the heavy graph requirements for resilient consensus \cite{yuan2021resilient,su2017reaching,sakavalas2020asynchronous}. Such techniques allow agents to communicate with non-direct neighbors through intermediate agents' relays \cite{goldsmith2005wireless}. On the other hand, quantized consensus has been motivated by the concerns on limited memories and transmission bandwidth on nodes in wireless sensor networks \cite{goldsmith2005wireless}.
Another vital motivation is that it can achieve exact consensus among agents. In contrast, approximate consensus usually guarantees that agents achieve consensus within a bounded interval.
There have been various works studying quantized consensus without adversaries \cite{chamie2016design,aysal2008distributed} and the case with adversaries.
The work \cite{dibaji2018resilient} studied quantized consensus under the malicious model and proved a tight graph condition for a one-hop MSR algorithm to achieve the objective.

A highly related topic named \textit{exact Byzantine consensus (EBC)}\footnote{Throughout this paper, EBC refers to binary consensus under the Byzantine model unless otherwise specified. Here, agents have only binary states $\{0,1\}$.} is popular and historical in distributed computing \cite{Lynch}. The works \cite{dolev1982byzantine} and \cite{tseng2015fault} have proposed tight graph conditions for synchronous EBC in undirected and directed networks, respectively. However, in the real world, delays are almost natural in agents' communication.
It is important to analyze whether an algorithm can succeed under asynchronous updates with delays. 
Interestingly, \cite{fischer1985impossibility} has shown that there exists no deterministic algorithm solving EBC in an asynchronous distributed system even in the presence of a single crash node. However, randomization brings a new perspective to this problem \cite{benor1983another}.
Recently, the work \cite{wang2020asynchronous} studied EBC under asynchronous updates with delays in undirected networks by introducing a randomization process in their algorithm. They derived the same necessary and sufficient condition as the one in \cite{dolev1982byzantine,tseng2015fault}. In this paper, we also implement a randomized quantizer on each agent, which is critical for our algorithm.

There exist other approaches for tackling the multi-valued asynchronous (exact) Byzantine agreement in complete networks through external authentication processes, where normal agents reach consensus on certain agent’s value \cite{cachin2001secure,abraham2019asymptotically}. Such authentication enables agents to verify the sender of a message through the digital signature technique provided by the external authority.
To the best of our knowledge, the problem of asynchronous quantized consensus under the Byzantine model in incomplete directed networks without authentication processes is still open, and we will fill this gap in Section~\ref{secbyzan}.
Related works for resilient quantized consensus (including EBC) are summarized in Table~\ref{table1}.

\begin{table}[t]\label{table1}
	\caption{Related works for resilient quantized consensus.}
	\renewcommand\arraystretch{1.5}
	\centering
	\begin{tabular}{c|c|c|c|c|c} 
		\hline 
		
		\multicolumn{2}{c|}{}&\multicolumn{2}{c|}{Synchronous}&\multicolumn{2}{c}{Asynchronous}\\  
		\hline
		\multirow{2}*{Malicious}&Undirected $\mathcal{G}$&\multicolumn{2}{c|}{\multirow{2}*{\cite{dibaji2018resilient, khan2020exact}, [$\ast$]}}&\multicolumn{2}{c}{\multirow{2}*{\cite{dibaji2018resilient}, [$\ast$]}}\\  
		\cline{2-2}  
		&Directed $\mathcal{G}$&\multicolumn{2}{c|}{ }&\multicolumn{2}{c}{ }  \\
		\cline{1-6}
		\multirow{2}*{Byzantine}&Undirected $\mathcal{G}$&\multicolumn{2}{c|}{ \cite{dolev1982byzantine}, [$\ast$]}&\multicolumn{2}{c}{ \cite{wang2020asynchronous}, [$\ast$]}\\
		\cline{2-6}
		&Directed $\mathcal{G}$&\multicolumn{2}{c|}{\cite{tseng2015fault}, [$\ast$]}&\multicolumn{2}{c}{ \ [$\ast$]}\\
		\hline
		\multicolumn{6}{c}{\makecell[l]{Note that [$\ast$] represents this work.}}
	\end{tabular}
	\vspace*{-5mm}
\end{table}

We summarize our contributions as follows. First, we develop a novel quantized multi-hop weighted MSR (QMW-MSR) algorithm to tackle resilient quantized consensus.
Utilizing a unified analysis, we prove necessary and sufficient conditions for our algorithm to succeed with synchronous/asynchronous updates under the malicious and Byzantine attack models. Moreover, our algorithm is fully distributed and our theoretical results hold for binary consensus if agents' initial states are restricted to binary values.


For the malicious model, we provide graph conditions tighter than the ones in \cite{dibaji2018resilient} for both synchronous and asynchronous updates.
Compared to the synchronous binary consensus work \cite{khan2020exact}, our algorithm is more efficient on exploiting the multi-hop techniques. In particular, we use less relay hops to achieve the same level of tolerance of malicious agents compared to the flooding-based\footnote{Flooding means that each node sends its values to the entire network.} algorithm \cite{khan2020exact}. See the examples in Section~\ref{sec_sim}. 
For the Byzantine model, we mainly compare our results with EBC works using flooding-based algorithms \cite{tseng2015fault,wang2020asynchronous}. Our algorithm studies the general multi-hop case, which is more flexible and lightweight.
Similarly, we can achieve the same level of tolerance of Byzantine agents with less relay hops in comparison with \cite{tseng2015fault,wang2020asynchronous}.

The rest of this paper is organized as follows. 
Section~II outlines the system model. Section~III presents the graph notions.
Sections~IV and V derive conditions under which our algorithm guarantees quantized consensus under synchronous and asynchronous updates.
Section~VI provides examples to verify the efficacy of our algorithm.
Lastly, Section~VII concludes the paper. 
A preliminary version of this work appeared in \cite{yuan2023resilient}. The current paper contains additional results for the Byzantine model and results for asynchronous networks.

\section{Preliminaries and Problem Formulation}

Consider the directed graph $\mathcal{G} = (\mathcal{V},\mathcal{E})$ consisting of the node set $\mathcal{V}=\{1,...,n\}$ and the edge set $\mathcal{E}\subset \mathcal{V} \times \mathcal{V}$. Here, the edge $(j,i)\in \mathcal{E}$ indicates that node $i$ can directly receive information from node $j$. 
An $l$-hop path from source $i_1$ to destination $i_{l+1}$ is a sequence of distinct nodes $(i_1, i_2, \dots, i_{l+1})$, where $(i_j, i_{j+1})\in \mathcal{E} $ for $j=1, \dots, l$. 
Node $i_{l+1}$ is said to be reachable from node $i_1$. 
The subgraph of $\mathcal{G} $ induced by node set $\mathcal{H}\subseteq\mathcal{V}$ is $\mathcal{G}_\mathcal{H}=(\mathcal{V}(\mathcal{H}),\mathcal{E}(\mathcal{H}))$, where $\mathcal{V}(\mathcal{H})=\mathcal{H}$ and $\mathcal{E}(\mathcal{H})=\{(i,j)\in \mathcal{E}: i,j\in \mathcal{H}\}$. Denote by $\mathbb{R}$, $\mathbb{Z}$, and $\mathbb{Z}_+$, respectively, the sets of real numbers, integers, and nonnegative integers. Let $\mathcal{N}_i^{l-}$ be the set of nodes that can reach node $i$ via paths of at most $l$ hops.
Also, let $\mathcal{N}_i^{l+}$ be the set of nodes that are reachable from node $i$ via paths of at most $l$ hops. Node $i$ is included in both sets above.
The $l$-th power of $\mathcal{G}$, denoted by $\mathcal{G}^l$, is a multigraph with the same nodes as $\mathcal{G}$ and an edge from node $j$ to node $i$ is defined by a path of length at most $l$ hops from $j$ to $i$ in $\mathcal{G}$. The adjacency matrix $A = [a_{ij} ]$ of $\mathcal{G}^l$ is given by $\alpha \leq a_{ij}<1$ if $j\in \mathcal{N}_i^{l-}$ and otherwise $a_{ij} = 0$, where $\alpha > 0$ is a fixed lower bound. We assume that $\sum_{j=1,j\neq i}^{n} a_{ij}\leq 1$. Let $L = [b_{ij} ]$ be the Laplacian matrix of $\mathcal{G}^l$, where $b_{ii} =\sum_{j=1,j\neq i}^{n}a_{ij}$ and $b_{ij} = -a_{ij}, i\neq j$.

	We introduce our message relaying model. 
	At time $k\geq 0$, each node $i_1$ exchanges the message tuples $m_{i_1i_{l+1}}[k]=(x_{i_1}[k],P_{i_1i_{l+1}}[k])$ consisting of its state $x_{i_1}[k]$ along each path $P_{i_1i_{l+1}}[k]$ with its $l$-hop neighbor $i_{l+1}$ via the relaying process in \cite{yuan2021resilient}. In particular, when source $i_1$ sends out $m_{i_1i_{l+1}}[k]$, $P_{i_1i_{l+1}}[k]$ is a vector of length $l+1$ with the source being $i_1$ and other entries being empty. Then, the one-hop neighbor $i_2$ receives $m_{i_1i_{l+1}}[k]$ from $i_1$, and it stores $x_{i_1}[k]$ for consensus and relays $m_{i_1i_{l+1}}[k]$ to its out-neighbors only by filling $i_2$ in the second entry of $P_{i_1i_{l+1}}[k]$. Such relaying continues until $m_{i_1i_{l+1}}[k]$ reaches node $i_{l+1}$. Denote by $\mathcal{V}(P)$ the set of nodes in $P$.

\subsection{Quantized Consensus and Update Rule}\label{problemsetting} 
Consider a time-invariant directed network $\mathcal{G} = (\mathcal{V},\mathcal{E})$.
The node set $\mathcal{V}$ is partitioned into the set of normal nodes $\mathcal{N}$ and the set of adversary nodes $\mathcal{A}$, where $|\mathcal{N}|=n_N$ and $|\mathcal{A}|=n_A$. The set $\mathcal{A}$ is always unknown to the normal nodes. At each time $k\geq 0$, each node $i$ first obtains the values of neighbors within $l$ hops, then it updates its value.
When there is no attack in the network, we employ the common real-valued consensus update rule extended from the one in \cite{bullo2009distributed}, given as
\begin{equation}\label{m1}
	\begin{aligned}
		x[k+1]&=x[k] +  u[k],   \\
		u[k]&=-L[k]x[k],
	\end{aligned}
\end{equation}
where $x[k]\in \mathbb{R}^n$ and $u[k]\in \mathbb{R}^n$ are the state vector and control input vector, respectively, and $L[k]$ is the Laplacian matrix of the $l$-th power graph $\mathcal{G}^l$ determined by the messages $m_{ij}[k], i\in \mathcal{V}$ and $j\in \mathcal{N}_i^{l-}$. Using \eqref{m1}, consensus is possible if $\mathcal{G}^l$ has a rooted spanning tree; see, e.g., \cite{bullo2009distributed}.

In this note, we focus on quantized consensus using the quantization function
$Q:\mathbb{R}\rightarrow \mathbb{Z}$ to transform the real-valued input in \eqref{m1} to integers \cite{aysal2008distributed,dibaji2018resilient}. Hence, the values and the inputs are constrained as $x_i [k] \in \mathbb{Z} $, $u_i [k] \in \mathbb{Z} $, $\forall i\in \mathcal{V} $. The quantization function is randomization based and is given by
\begin{equation}\label{quantizer}
	Q(y)=\left\{
	\begin{array}{lll} 
		\lfloor y\rfloor &\textup{with probability} \medspace\medspace\medspace p(y),\\
		\lceil y\rceil & \textup{with probability} \medspace\medspace\medspace 1-p(y),
	\end{array}
	\right.
\end{equation}
where $p(y)=\lceil y\rceil-y$, $\lfloor \cdot \rfloor$ and $\lceil \cdot \rceil$ denote the floor and the ceiling functions, respectively. Then, based on \eqref{m1}, we can write the quantized control input for normal node $i$ as
\begin{equation}\label{quantized_input}
	u_i[k]=Q\bigg( \sum_{j\in \mathcal{N}_i^{l-}} a_{ij}[k]x_j[k]\bigg) ,  
\end{equation}
where $a_{ij} [k]$ is the $(i, j)$th entry of the adjacency matrix $A[k]$ of graph $\mathcal{G}^l$ at time $k$. Then we denote by $x^N [k] \in \mathbb{Z}^{n_N} $ and $x^A [k] \in \mathbb{Z}^{n_A} $ the state vectors of normal nodes and adversary nodes, respectively.
Moreover, the probability $p(\cdot)$ in \eqref{quantizer} can be different on each node at each time as long as it is in $(0,1)$. Thus, \eqref{quantized_input} can be implemented in a distributed fashion.

Denote by $\mathcal{U}[k] \subset \mathcal{V} $ the set of agents updating at time $k$. The system is said to be synchronous if $\mathcal{U}[k] = \mathcal{V} $, $\forall k$, and otherwise it is asynchronous \cite{dibaji2018resilient,senejohnny2019resilience}.

\subsection{Attack Models}\label{threatmodel}

We introduce the attack models studied in this note, which are the multi-hop versions of the ones in \cite{leblanc2013resilient}, \cite{vaidya2012iterative}.

\begin{definition}
	\textit{($f$-total / $f$-local set)}
	The set of adversary nodes $\mathcal{A}$ is said to be $f$-total
	if $| \mathcal{A}| \leq f$.
	Similarly, it is said to be $f$-local if $\forall i\in \mathcal{N}$, $|\mathcal{N}_i^{l-} \cap \mathcal{A}| \leq f$.
\end{definition}

\begin{definition}
	\textit{(Malicious / Byzantine nodes)}
	An adversary node $i\in \mathcal{A}$ is said to be malicious
	if it arbitrarily modifies its own value and relayed values,\footnote{
		An adversary node may also decide not to transmit any value \cite{Lynch}.} but sends the same state
	value and same relayed values to its neighbors at each iteration. Moreover, a Byzantine node sends different state values and different relayed values to its neighbors at each iteration.
\end{definition}

The malicious model is reasonable in applications such as wireless sensor networks and robotic networks, where neighbors' information is obtained by broadcast communication \cite{goldsmith2005wireless}. The Byzantine model is generally assumed in point-to-point networks and is more adversarial given that all malicious nodes are Byzantine, but not vice versa \cite{Lynch}, \cite{leblanc2013resilient}.

Following the literature \cite{leblanc2013resilient,su2017reaching}, we make the assumption that each normal node has the knowledge of $f$ and the topology information of the graph within $l$ hops.
Moreover, we introduce the following assumption from \cite{yuan2021resilient,su2017reaching} for analytical simplicity. Yet, manipulating message paths can be easily detected; see the discussions in \cite{su2017reaching, yuan2021resilient}.

\begin{assumption}
	Each adversary node $i$ can manipulate its own state $x_i[k]$ and the values in the messages that they relay, but cannot change the path values $P$ in such messages. 
\end{assumption}

\subsection{Resilient Quantized Consensus and Algorithm}

We now introduce the type of consensus among the normal agents to be sought in this note, which is also studied in \cite{dibaji2018resilient}.

\begin{definition}
	If for any possible sets and behaviors of the
	adversary agents and any state values of the normal
	nodes, the following two conditions are satisfied,
	then we say that the normal agents reach 
	resilient quantized consensus:
	\begin{enumerate}
		\item Safety: There exists a bounded safety interval $\mathcal{S}$ determined by the initial values of the normal agents such that $x_i[k] \in \mathcal{S}, \forall i \in \mathcal{N}, k \in \mathbb{Z}_+$. 
		\item Agreement: 
		There exists a finite time $k_a \geq 0$ such that Prob$\{x^N [k_a ]\in  \mathcal{C} : x[0]\} =1 $, where the consensus set $\mathcal{C}$ is defined as
		$\mathcal{C}=\{ x\in \mathbb{Z}^{n_N} :  x_1=\cdots=x_{n_N}  \}.$
	\end{enumerate}
\end{definition}

Next, we introduce our resilient Algorithm~1. It is the quantized version of the MW-MSR algorithm in \cite{yuan2021resilient}, where the notion of minimum message cover (MMC) is crucial.

\begin{definition} For a graph $\mathcal{G} = (\mathcal{V},\mathcal{E})$, let $\mathcal{M}$ be a set of messages transmitted through $\mathcal{G}$, and let $\mathcal{P}(\mathcal{M})$ be the set of message paths of all the messages in $\mathcal{M}$, i.e., $\mathcal{P}(\mathcal{M}) =\{\mathrm{path}(m):m \in \mathcal{M}\}$. A \textit{message cover} of $\mathcal{M}$ is a set of nodes $\mathcal{T}(\mathcal{M})\subset \mathcal{V}$ whose removal disconnects all message paths, i.e., for each path $P\in \mathcal{P}(\mathcal{M})$, we have $\mathcal{V}(P)\cap \mathcal{T}(\mathcal{M})\neq \emptyset$. In particular, a \textit{minimum} message cover of $\mathcal{M}$ is defined by
	\begin{equation*}
		\mathcal{T}^*(\mathcal{M})\in	\arg \min_{\substack{ \mathcal{T}(\mathcal{M}): \textup{ Cover of } \mathcal{M}}} 	\left|  \mathcal{T} (\mathcal{M})\right| . 
	\end{equation*}
\end{definition}
\vspace{1mm}


	We briefly explain the trimming and a detailed example can be found in \cite{yuan2021resilient}.
	In step 2(b), node $i$ checks the first $f+1$ largest values of $\overline{\mathcal{M}}_i[k]$, and if the cardinality of their MMC is no more than $f$, then it checks the first $f+2$ values of $\overline{\mathcal{M}}_i[k]$. It continues until for the first $f+h$ (with $h\geq 1$) values, the cardinality of their MMC is $f+1$. Then $\overline{\mathcal{R}}_i[k]$ is taken as the first $f+h-1$ values of $\overline{\mathcal{M}}_i[k]$.
	In step 3, it uses only the safe values in $\mathcal{M}_i[k]\setminus \mathcal{R}_i[k]$. Here, the value originating from the same node via different paths may be used multiple times, which is associated with a larger weight in \eqref{msrupdate}.

\begin{algorithm}[t]
	\caption{Synchronous QMW-MSR Algorithm     }
	\LinesNumbered 
	\KwIn{Node $i$ knows $x_i[0]$, $\mathcal{N}_i^{l-}$, $\mathcal{N}_i^{l+}$. }
	
	\For{$k\geq0$}{
		
		\SetKwBlock{newbox}{1) Exchange messages:}{}
		\newbox{
			\SetAlgoVlined
			Send $m_{ij}[k]=(x_i[k],P_{ij}[k])$ to $\forall j\in \mathcal{N}_i^{l+}$. 
			
			Receive $m_{ji}[k]=(x_j[k],P_{ji}[k])$ from $\forall j\in \mathcal{N}_i^{l-}$ and store them in the set $\mathcal{M}_i[k]$.
			
			Sort $\mathcal{M}_i[k]$ in an increasing order based on the message values (i.e., $x_j[k]$ in $m_{ji}[k]$).
		}
		
		\SetKwBlock{newbox}{2) Remove extreme values:}{}
		\newbox{
			\SetAlgoVlined
			
			(a) Define two subsets of $\mathcal{M}_i[k]$:
			\begin{equation*}
				\overline{\mathcal{M}}_i[k]=\{ m\in \mathcal{M}_i[k]: \mathrm{value}(m)> x_i[k]  \},
			\end{equation*}
			\begin{equation*}
				\underline{\mathcal{M}}_i[k]=\{ m\in \mathcal{M}_i[k]: \mathrm{value}(m)< x_i[k]  \}.
			\end{equation*}
			
			(b) Get $\overline{\mathcal{R}}_i[k]$ from $\overline{\mathcal{M}}_i[k]$:
			
			\vspace{1mm}
			
			\eIf{   $\left|  \mathcal{T}^* (\overline{\mathcal{M}}_i[k])\right| <f$ }{
				
				\vspace{1mm}
				$\overline{\mathcal{R}}_i[k]=\overline{\mathcal{M}}_i[k]$;
			}
			{
				Choose $\overline{\mathcal{R}}_i[k]$ s.t. (i)
				$\forall m\in \overline{\mathcal{M}}_i[k]\setminus \overline{\mathcal{R}}_i[k]$, $\forall m'\in \overline{\mathcal{R}}_i[k]$, $\mathrm{value}(m) \leq \mathrm{value}(m')  \medspace\medspace \medspace \medspace  $ 
				and (ii) $\left|  \mathcal{T}^* (\overline{\mathcal{R}}_i[k])\right| =f$. 
				
			}
			
			(c) Similar to (b), get $\underline{\mathcal{R}}_i[k]$ from $\underline{\mathcal{M}}_i[k]$, which contains smallest message values.
			
			(d) $\mathcal{R}_i[k]=\overline{\mathcal{R}}_i[k]\cup\underline{\mathcal{R}}_i[k]$.
		}
		
		\SetKwBlock{newbox}{3) Update:}{}
		\newbox{
			\SetAlgoVlined
			$\medspace\medspace \medspace\medspace\medspace\medspace\medspace\medspace a_{i}[k]=1/(\left| \mathcal{M}_i[k]\setminus \mathcal{R}_i[k] \right| )$,
			\begin{equation}
				x_i[k+1]=Q\bigg(\sum_{m\in \mathcal{M}_i[k]\setminus \mathcal{R}_i[k]} a_{i}[k] \medspace\medspace \mathrm{value}(m)\bigg).   \label{msrupdate}
			\end{equation}
			\vspace{-5.0mm}
		}
		\KwOut{$x_i[k+1]$.}
	}
	\vspace{-1.0mm}
\end{algorithm}

\section{Graph Robustness with Multi-Hop Communication}

\subsection{Robustness with Multi-Hop Communication}\label{sec_robustness}

We introduce the notions of robustness and strict robustness
with $l$ hops. The first one is a tight graph condition guaranteeing
resilient consensus using MSR algorithms \cite{leblanc2013resilient,yuan2021resilient}.

\begin{definition}\label{rs-robust}
	\textit{($(r,s)$-robustness with $l$ hops)}
	A directed graph $\mathcal{G} = (\mathcal{V},\mathcal{E})$ is said to be $(r,s)$-robust with $l$ hops with respect to a given set $\mathcal{F}\subset \mathcal{V}$,
	if for every pair of nonempty disjoint subsets $\mathcal{V}_\text{1},\mathcal{V}_\text{2}\subset \mathcal{V}$, at least one of the following conditions holds:
		\begin{enumerate}
				\item $\mathcal{X}_{\mathcal{V}_1}^r=\mathcal{V}_1$;
				\item $\mathcal{X}_{\mathcal{V}_2}^r=\mathcal{V}_2$;
				\item  $\left| \mathcal{X}_{\mathcal{V}_1}^r\right| +\left| \mathcal{X}_{\mathcal{V}_2}^r\right| \geq s$,
		\end{enumerate}
	
	\noindent where $\mathcal{X}_{\mathcal{V}_a}^{r}=\{i\in \mathcal{V}_a : |\mathcal{I}_{i, \mathcal{V}_a}^{\mathcal{F}}| \geq r \}$, ($a=1,2$). Here, $\mathcal{I}_{i, \mathcal{V}_a}^{\mathcal{F}}$ is the set of independent paths to node $i\in\mathcal{V}_a$ of at most $l$ hops originating from nodes outside $\mathcal{V}_a$ and all these paths do not have any nodes in set $\mathcal{F}$ as intermediate nodes.\footnote{In independent paths to node $i$, only node $i$ is common. Moreover, nodes in $\mathcal{F}$ can be source or destination nodes of these paths}.
	Moreover, if graph $\mathcal{G}$ satisfies this property with respect to any set $\mathcal{F}$ satisfying the $f$-total or $f$-local model, then we say that $\mathcal{G}$ is \textit{$(r,s)$-robust with $l$ hops} (under the $f$-total or $f$-local model). If $\mathcal{G}$ is $(r,1)$-robust with $l$ hops, it is also \textit{$r$-robust with $l$ hops}.
\end{definition}

To satisfy $|\mathcal{I}_{i, \mathcal{V}_a}^{\mathcal{F}}| \geq r$, there should be at least $r$ source nodes outside $\mathcal{V}_a$ and at least one independent path of length at most $l$ hops from each of the $r$ source nodes to node $i$.

To evaluate the robustness with $l\geq 2$ hops of a given graph, we need to check over all the possible subcases of different $\mathcal{F}$ satisfying the $f$-total/local model. For example, 
the cycle graph in Fig.~\ref{graph1}(a) is not $(2, 2)$-robust with $l\leq 3$ hops under the 1-total/local model. 
Consider the node sets $\mathcal{V}_1=\{1,2,3,4,5,6\}$ and $\mathcal{V}_2=\{7,8\}$ with $\mathcal{F}=\{8\}$. When $l\leq 3$, observe that $\mathcal{X}_{\mathcal{V}_1}^{2}= \emptyset$ and $\mathcal{X}_{\mathcal{V}_2}^{2} \cap \{7\}= \emptyset$. Hence, the cycle graph in Fig.~\ref{graph1}(a) is not $(2, 2)$-robust with $l\leq 3$ hops. However, for the same sets, when $l=4$, observe that  $\mathcal{X}_{\mathcal{V}_1}^{2}= \{3,4\}$. In fact, one can check all the combinations of node subsets, and the cycle graph meets the condition for being $(2, 2)$-robust with $4$ hops.

Moreover, as the number of hops $l$ increases, the robustness (and strict robustness defined next) of a graph is generally nondecreasing. 
The improvement of (strict) robustness by increasing $l$ depends on the graph structure and may differ in different graphs. See \cite{yuan2021resilient} for more graph properties of robustness with $l\geq 1$ hops.

\subsection{Strict Robustness with Multi-Hop Communication}

In \cite{yuan2022asynchronous}, we defined the notion of strict robustness with $l$ hops to characterize the necessary and sufficient condition for achieving approximate Byzantine consensus. As \cite{vaidya2012iterative,su2017reaching}, our notion is defined on the original network $\mathcal{G}$ instead of the \textit{normal network} (i.e., the subgraph of normal nodes). Hence, our notion can be verified prior to the algorithm deployment and without knowing the identity of Byzantine nodes.

\begin{figure}[t]
	\centering
	\subfigure[]{
		\includegraphics[width=1.03in]{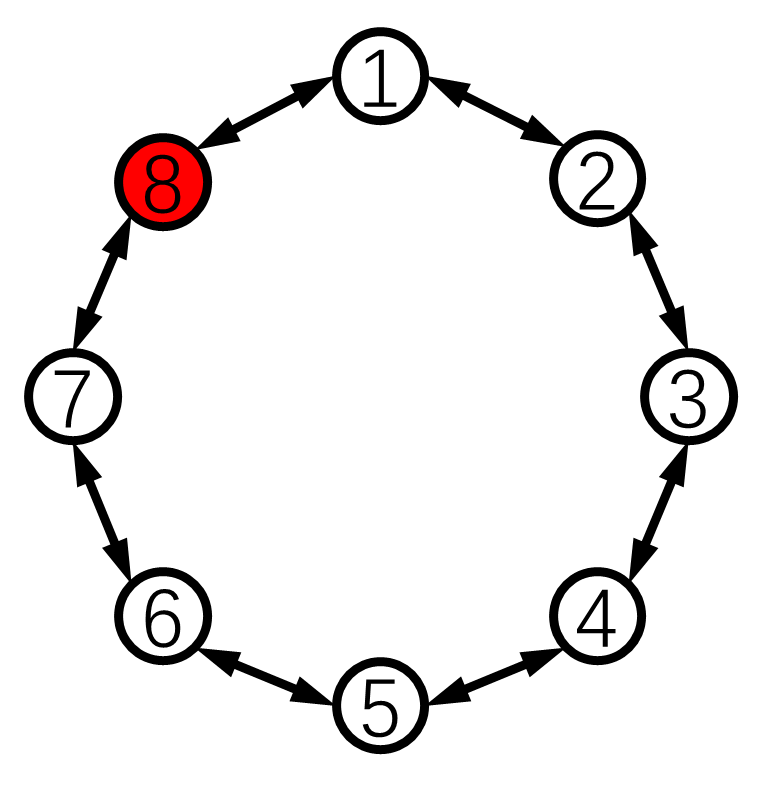}
	}
	\quad
	\subfigure[]{
		\includegraphics[width=1.12in]{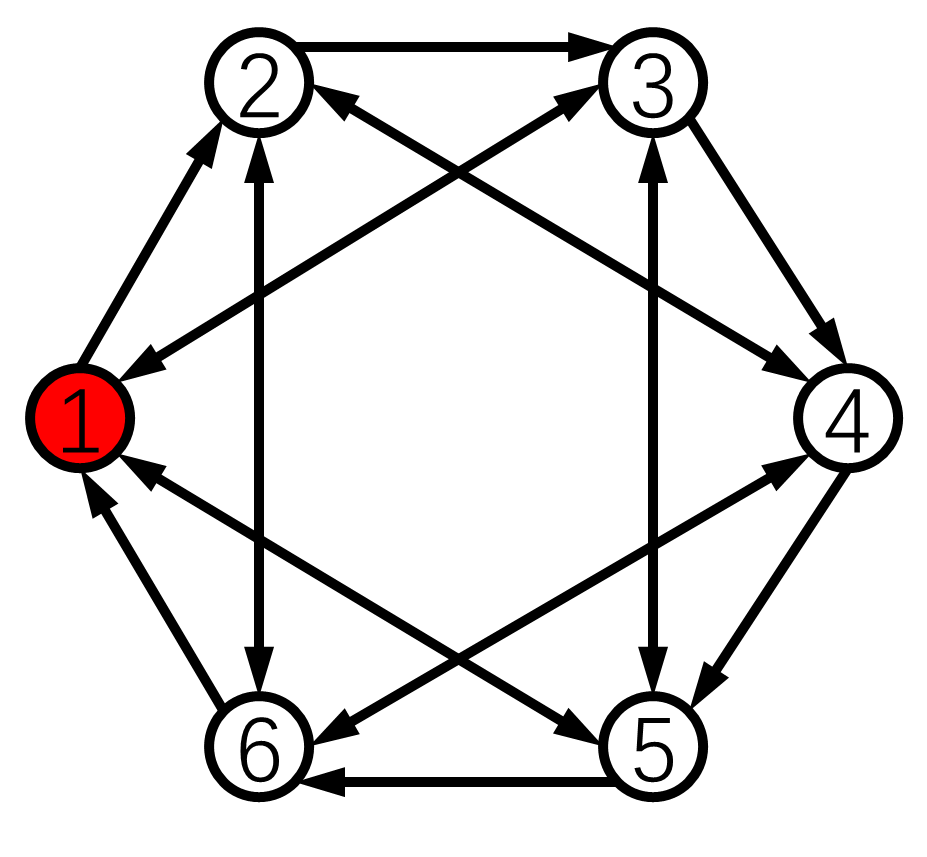}
	}
	\vspace{-4pt}
	\caption{(a) The cycle graph is not (2, 2)-robust with $l\leq 3$ hops but is (2,2)-robust with 4 hops. (b) The graph is 2-strictly robust with 2 hops and hence is (2,2)-robust with 2 hops. }
	\label{graph1}
	\vspace*{-3.0mm}
\end{figure}

\begin{definition}\label{strict_robustness}	
	\textit{($r$-strict robustness with $l$ hops)}
	Let $\mathcal{F} \subset \mathcal{V}$ and denote the subgraph of $\mathcal{G}=(\mathcal{V},\mathcal{E})$ induced
	by node set $\mathcal{H}=\mathcal{V}\setminus\mathcal{F}$ as $\mathcal{G}_{\mathcal{H}}$.
	Then, graph $\mathcal{G}$ is said to be $r$-strictly robust with $l$ hops with respect to $\mathcal{F}$ if subgraph $\mathcal{G}_{\mathcal{H}}$ is $r$-robust with $l$ hops.\footnote{The removed set $\mathcal{F}$ is still used to count the robustness of the remaining graph $\mathcal{G}_{\mathcal{H}}$ since the strict robustness is a property of the original graph $\mathcal{G}$.}
	If graph $\mathcal{G}$ satisfies this property with respect to any $f$-total or $f$-local set $\mathcal{F}$, then we say that $\mathcal{G}$ is $r$-strictly robust with $l$ hops (under the $f$-total or $f$-local model).
\end{definition}

We illustrate the idea of how to determine the strict robustness with $l$ hops of a given graph. The graph in Fig.~\ref{graph1}(b) is not $2$-strictly robust with $1$ hop. We check one subcase and other subcases follow a similar procedure. Consider the node sets $\mathcal{V}_1=\{3,5\}$ and $\mathcal{V}_2=\{2,4,6\}$ with $\mathcal{F}=\{1\}$. When $l=1$, if we remove set $\mathcal{F}$, then $\mathcal{X}_{\mathcal{V}_1}^{2}= \emptyset$ and $\mathcal{X}_{\mathcal{V}_2}^{2}= \emptyset$. Apparently, the remaining graph is not $2$-robust with $1$ hop.
However, when $l=2$, we have $\mathcal{X}_{\mathcal{V}_1}^{2}= \mathcal{V}_1$. Hence, $\mathcal{V}_1$ and $\mathcal{V}_2$ satisfy the condition for $2$-robustness with $2$ hops. After checking all the combinations of node subsets, we can conclude that this graph is $2$-strictly robust with $2$ hops, and hence according to Lemma~\ref{abc}, it is also  $(2,2)$-robust with $2$ hops.

The following lemma from \cite{yuan2022asynchronous} characterizes the topology gaps between the three graph conditions presented in related resilient consensus works \cite{leblanc2013resilient,yuan2021resilient,dibaji2018resilient,yuan2022asynchronous}.

\begin{lemma}\label{abc}
	For the following conditions on any directed graph $\mathcal{G} = (\mathcal{V},\mathcal{E})$ under the $f$-total or $f$-local model ($l\geq1$):
	
	$(A)$ $\mathcal{G}$ is $(2f+1)$-robust with $l$ hops,
	
	$(B)$ $\mathcal{G}$ is $(f+1)$-strictly robust with $l$ hops,
	
	$(C)$ $\mathcal{G}$ is $(f + 1,f+1)$-robust with $l$ hops,
	
	\noindent  it holds that $(A) \Rightarrow (B)$ and $(B) \Rightarrow (C)$. Moreover, $(C)\nRightarrow (B)$ and $(B)\nRightarrow (A)$.
\end{lemma}

Condition $(A)$ is a sufficient condition for the $f$-total/local malicious model with communication delays \cite{dibaji2018resilient,yuan2021resilient}. Condition $(B)$ is a necessary and sufficient condition for the $f$-total/local Byzantine model with/without delays \cite{yuan2022asynchronous}. Condition $(C)$ is a necessary and sufficient condition for the $f$-total malicious model without delays \cite{leblanc2013resilient,yuan2021resilient}.

\subsection{The Case of Unbounded Path Length} \label{sec_unbounded}

In this subsection, we discuss the relations between the graph conditions in this paper and the ones in the recent works for binary consensus \cite{khan2020exact,tseng2015fault} shown in Table~\ref{table1}. The authors of \cite{khan2020exact} studied EBC under the local broadcast model, which is equivalent to the $f$-total malicious model. Their algorithm is based on a noniterative flooding algorithm, where nodes must relay their values over the entire network. This model corresponds to the case of unbounded path length in our work, i.e. $l\geq l^*$, where $l^*$ is the longest cycle-free path length of the network.
In our previous work \cite{yuan2021resilient} studying real-valued consensus, we have proved that our graph condition (i.e., $(f + 1,f+1)$-robustness with $l$ hops) is equivalent to theirs when $l\geq l^*$. Besides, to achieve the same tolerance of malicious agents as the algorithm in \cite{khan2020exact}, our algorithm usually requires less than $l^*$-hop relaying for general graphs.

In our previous work \cite{yuan2022asynchronous} studying approximate Byzantine consensus, we have proved that a graph is $(f+1)$-strictly robust with $l$ hops if and only if it satisfies the graph condition in \cite{su2017reaching}, which studied synchronous Byzantine consensus with $l$-hop communication. In \cite{su2017reaching}, it has been reported that the graph condition there for synchronous Byzantine consensus is the same as the one in \cite{tseng2015fault} with $l\geq l^*$. Hence, we have established the equivalence between our condition and the condition in \cite{tseng2015fault} for the case of unbounded path length.

\subsection{Classes of Graphs Satisfying Robustness Conditions}\label{sec_construction}

We present several classes of graphs which meet the topological conditions in our results. We note that it is computationally complex to check robustness for a given graph since combinatorial processes are involved (see, e.g., \cite{leblanc2013resilient,usevitch2020determining}). With multi-hop communication, we are able to improve robustness of certain classes of graphs and provide analytic results. In what follows, we discuss graphs that meet robustness conditions with $l$ hops, which is about half or even much smaller than $l^*$. Hence, they require less relay hops than the flooding algorithms in \cite{wang2020asynchronous,khan2020exact} for the same level of robustness. Detailed comparisons with these works are presented in Sections~\ref{secsyn} and \ref{secbyzan}.

To begin with, we introduce a lemma for cycle graphs from \cite{yuan2021resilient}. A cycle graph $\mathcal{C}_n$ with $n$ nodes may be very fragile in terms of robustness because each node has only two neighbors. Thus, when $n>3$, resilient consensus is impossible by one-hop MSR algorithms even when only one of the neighbors becomes malicious. However, the following lemma shows that $\mathcal{C}_n$ can tolerate one malicious node if each node is capable to send data over $\lceil l^*/2 \rceil$ hops, where $l^*=n-1$. See also the result in Theorem~\ref{syn} to be presented.

\begin{figure}[t]
	\centering
	\subfigure[]{
		\includegraphics[width=1.2in]{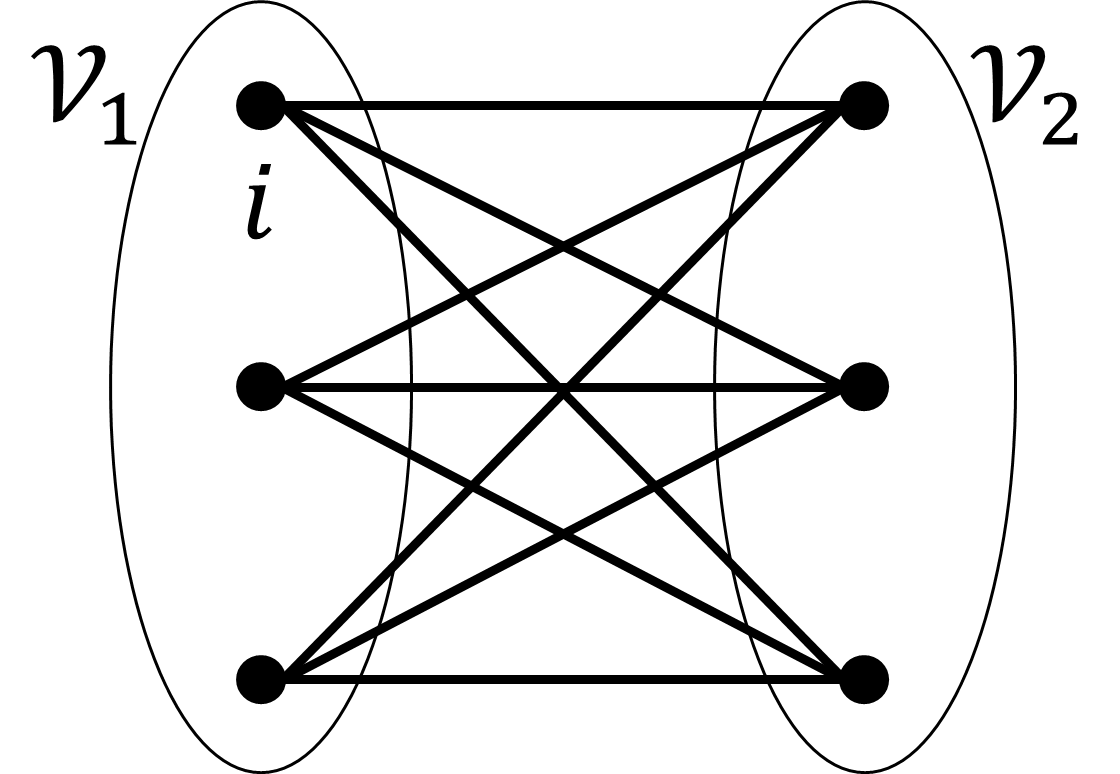}
	}
	\quad
	\subfigure[]{
		\includegraphics[width=1.05in]{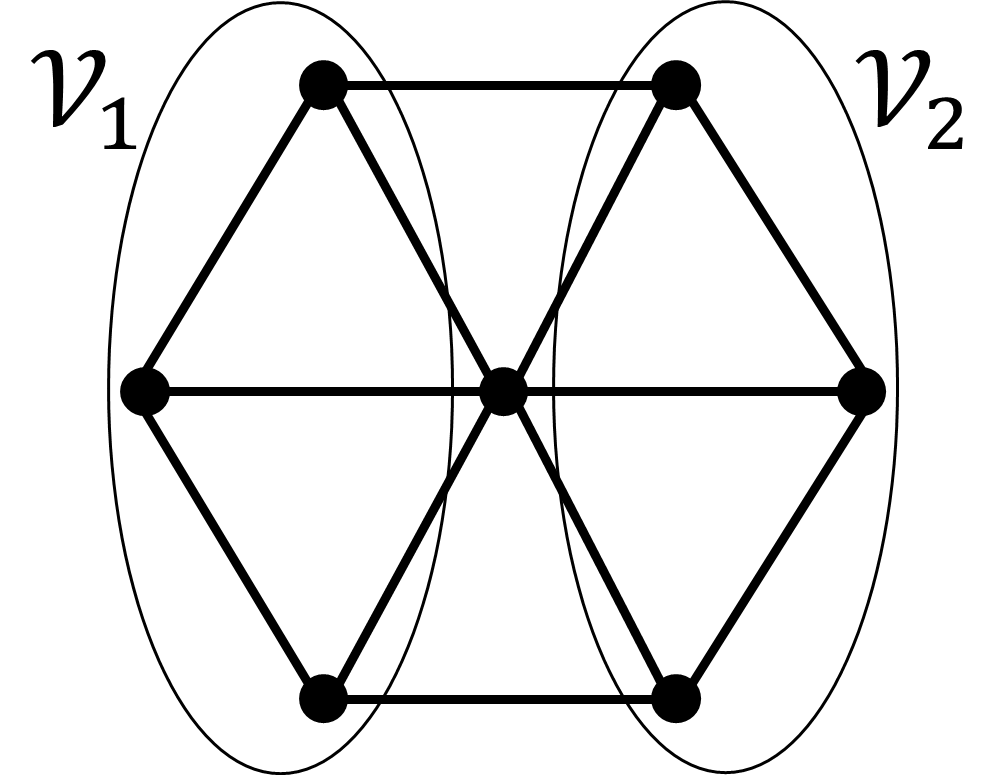}
	}
	\vspace{-4pt}
	\caption{Illustration for graphs satisfying our conditions. (a) Complete bipartite graph $\mathcal{K}_{3,3}$ having robustness with $l$ hops. (b) Wheel graph $\mathcal{W}_7$ having strict robustness with $l$ hops.}
	\label{example_graphs}
	\vspace*{-3.0mm}
\end{figure}

\begin{lemma}\label{lemmacycle}
	The cycle graph $\mathcal{C}_n$ with $n>2$ nodes is $(2,2)$-robust with $\lceil l^*/2 \rceil$ hops under the $1$-total model.
\end{lemma}

We next consider the class of complete bipartite graphs \cite{west2001introduction}. Consider the undirected graph $\mathcal{G} = (\mathcal{V},\mathcal{E})$ with two nonempty and disjoint node subsets $\mathcal{V}_1$ and $\mathcal{V}_2$ that partition $\mathcal{V}$ as $ \mathcal{V} = \mathcal{V}_1 \cup \mathcal{V}_2$, where $|\mathcal{V}_1|=n_1$ and $|\mathcal{V}_2|=n_2$. A graph is a bipartite graph if each node in $\mathcal{V}_1$ has edges to only nodes in $\mathcal{V}_2$ and vice versa. Moreover, it is a complete bipartite graph denoted by $\mathcal{K}_{n_1,n_2}$ if each node in $\mathcal{V}_1$ has edges to all nodes in $\mathcal{V}_2$ and vice versa (see Fig.~\ref{example_graphs}(a)). For $\mathcal{K}_{n_1,n_2}$, it holds $l^*=n-1$ as well. Let $d=\min\{n_1,n_2\}$.
 
In a complete bipartite graph, with two-hop communication, each node can reach all other nodes in the network, which is a strong feature. For such a graph, we have the following result.
 
\begin{lemma}\label{lemma2f}
	The complete bipartite graph $\mathcal{K}_{n_1,n_2}$ is $(\lfloor \frac{d }{2} \rfloor +1 , \lfloor \frac{d }{2} \rfloor +1)$-robust with $l\geq 2$ hops under the $\lfloor \frac{d }{2} \rfloor$-total model.
\end{lemma}

Its proof can be found in the Appendix~A. By Lemma~\ref{lemma2f}, we could generate desirable graphs for our algorithm, which can tolerate any number $f\leq \lfloor \frac{d }{2} \rfloor$ of malicious agents by introducing just two-hop communication. Relevant results are given in Section~\ref{secsyn}.

Furthermore, we introduce a class of graphs that meet the condition of strict robustness with $l$ hops.
A wheel graph $\mathcal{W}_n$ consists of a cycle subgraph and a centering node connected with all other nodes as shown in Fig.~\ref{example_graphs}(b), see also \cite{west2001introduction}.
In such a graph too, two-hop relays are enough for any two nodes to communicate whereas $l^*=n-1$.
The proof of Lemma~\ref{lemmastarcyle} can be found in the Appendix~B.

\begin{lemma}\label{lemmastarcyle}
	The wheel graph $\mathcal{W}_n$ with $n>3$ nodes is $2$-strictly robust with $\lfloor l^*/4 \rfloor+1$ hops under the $1$-total model.
\end{lemma}

It is interesting that compared to cycle graphs, wheel graphs have a more centralized structure with the centering node, and their strict robustness can be enhanced with less number of hops. However, in both cases, the maximum tolerable number of adversary agents is one. This is indicated by the $1$-total model in Lemmas~\ref{lemmacycle} and \ref{lemmastarcyle}.

\section{Synchronous Network}\label{secsyn}

In this section, we analyze the performance of the synchronous QMW-MSR algorithm under the malicious model. All nodes update values synchronously, i.e., $\mathcal{U}[k] = \mathcal{V}, \forall k$.

Reorder the agents so that the normal agents take indices $1,\dots,n_N$ and the adversary agents take $n_N +1,\dots,n$. Then the state vector and control input vector can be written as 
\begin{equation} 
	x[k]=\begin{bmatrix} x^N[k]  \\ x^A[k] \end{bmatrix}, \medspace u[k]=\begin{bmatrix} u^N[k]  \\ u^A[k] \end{bmatrix}. 
\end{equation}
Regarding the control inputs $u^N[k]$ and $u^A[k]$, the normal agents follow \eqref{msrupdate} while the adversary agents may not. Hence, they can be expressed as
\begin{equation}
	\begin{array}{lll} 
		u^N[k] =Q\big(-L^N[k]x[k]\big),\\
		u^A[k] : \textup{arbitrary,}
	\end{array}
\end{equation}
where $L^N[k]\in \mathbb{R}^{n_N\times n}$ is the matrix formed by the first $n_N$ rows of $L[k]$ associated with normal agents. The row sums of this matrix $L^N[k]$ are zero as in $L[k]$. Thus, the system is expressed as
\begin{equation} \label{system1}
	x[k+1]=Q\bigg(\left(  I_n -   \begin{bmatrix} L^N[k] \\ 0 \end{bmatrix} \right)  x[k]\bigg)  +   \begin{bmatrix} 0  \\ I_{n_A} \end{bmatrix}u^A[k].
\end{equation}

In the following theorem, we characterize a necessary and sufficient graph condition for the malicious model. For the agents using Algorithm~1, the safety interval is given by 
\begin{equation}  \label{safety1}
	\mathcal{S}=\big[ \min x^N[0], \max x^N[0] \big].        	
\end{equation}

\begin{theorem}\label{syn}
	Consider a directed network $\mathcal{G} = (\mathcal{V},\mathcal{E})$ with $l$-hop communication, where each normal node updates its value according to the synchronous QMW-MSR algorithm with parameter
	$f$. Under the $f$-total malicious model, resilient quantized
	consensus is achieved almost surely with safety interval \eqref{safety1} if and only if $\mathcal{G}$ is $(f + 1, f + 1)$-robust with $l$ hops.
\end{theorem}

To establish quantized consensus in this probabilistic setting, we need the following lemma, which is sufficient for guaranteeing resilient quantized consensus almost surely \cite{dibaji2018resilient}.

\begin{lemma}\label{synlemma}
	Consider a directed network $\mathcal{G} = (\mathcal{V},\mathcal{E})$ using the QMW-MSR algorithm. Suppose that the following three conditions are satisfied $\forall i \in\mathcal{N}$:
	\begin{enumerate}
		\item[C1)] There exists a bounded set $\mathcal{S}$ determined by the initial states of the normal nodes such that
		$x_i[k] \in \mathcal{S}, \forall i \in \mathcal{N}, k \in \mathbb{Z}_+$. 
		
		\item[C2)] For each state $x[k]=x_0$ at time $k$,
		there exist a finite time $k_b$ such that Prob$\{x^N [k+k_b ]\in  \mathcal{C} \thinspace|\thinspace x[k]=x_0\} >0 $.
		
		\item[C3)] If $x^N[k]\in \mathcal{C}$, then $x^N[k']\in \mathcal{C}, \forall k'>k$.
	\end{enumerate}
	Then, the network reaches quantized consensus almost surely.
\end{lemma}

If (C1)--(C3) hold $\forall i \in\mathcal{N}$, then the scenarios for reaching quantized consensus occur infinitely often with high probability. This is because the probability for such an event to occur is positive by (C2). Then, once normal agents reach consensus, such consensus is preserved indefinitely by (C3).

\textit{Proof of Theorem \ref{syn}:}
(Necessity)
If $\mathcal{G}$ is not $(f +1, f +1)$-robust with $l$ hops, then
there are nonempty, disjoint subsets $\mathcal{V}_1, \mathcal{V}_2\subset \mathcal{V}$ such that the following conditions hold:
\begin{enumerate}
	\item $|\mathcal{X}_{\mathcal{V}_1}^{f+1}|<|\mathcal{V}_1|$;
	\item $|\mathcal{X}_{\mathcal{V}_2}^{f+1}|<|\mathcal{V}_2|$;
	\item  $| \mathcal{X}_{\mathcal{V}_1}^{f+1} | +\left| \mathcal{X}_{\mathcal{V}_2}^{f+1}\right| \leq f$.
\end{enumerate}

Suppose that
$x_i[0]=a, \medspace \forall i\in \mathcal{V}_1$, and $x_j[0]=b, \medspace \forall j\in \mathcal{V}_2$ with $a < b$, and $x_q[0]=\lfloor (a+b)/2\rfloor, \medspace \forall q\in \mathcal{V}\setminus (\mathcal{V}_1 \cup \mathcal{V}_2)$.
Under condition 3) above, suppose that all nodes in $\mathcal{X}_{\mathcal{V}_1}^{f+1}$ and $\mathcal{X}_{\mathcal{V}_2}^{f+1}$ are malicious and take constant values.
Then, there is still at least one normal node in both $\mathcal{V}_1$ and $\mathcal{V}_2$ due to conditions 1) and 2) above. Then these normal nodes in $\mathcal{V}_1$ and $\mathcal{V}_2$ remove all the values originating from the nodes outside their respective sets since the MMC of these values has a cardinality equal to $f$ or less. By Algorithm~1, such normal nodes will keep their values and consensus cannot be achieved.

(Sufficiency) We show that the conditions (C1)--(C3) in Lemma \ref{synlemma} hold. To prove (C1), define the maximum and minimum values of the normal nodes at
time $k$ as
\begin{align*}
	\overline{x}[k]  =\max x^N[k],  \medspace\medspace\medspace
	\underline{x}[k]  =\min x^N[k]. 
\end{align*}
Observe that $\forall i\in \mathcal{N}$, the values used in step 3 of Algorithm~1 always lie within the interval $\big[ \underline{x}[k], \overline{x}[k] \big]$. 
Moreover, the update rule in \eqref{system1} is a quantized convex combination of the values in $\big[ \underline{x}[k], \overline{x}[k] \big]$. Therefore, $x_i[k+1]\in \big[ \underline{x}[k], \overline{x}[k] \big], \forall i\in \mathcal{N}$, and by induction, we have $x_i[k]\in \mathcal{S},\forall i\in \mathcal{N}, k\in \mathbb{Z}_+$.     

Then, we prove (C2). 
From above, $\overline{x}[k]$ and $\underline{x}[k]$ are monotone and bounded sequences, and thus there is a finite time $k_c$ such
that they both reach their final values with probability 1.
Denote the final values of $\overline{x}[k]$ and $\underline{x}[k]$ by $\overline{x}^*$ and $\underline{x}^*$,
respectively.
We will prove by contradiction to show that $\overline{x}^*=\underline{x}^*$, thus consensus is reached.
Suppose that $\overline{x}^*>\underline{x}^*$. When $k\geq k_c$, we can define the following sets
\begin{equation*}
	\begin{aligned}
		\mathcal{Z}_1[k]&=\{i\in \mathcal{V}: x_i[k]\geq \overline{x}^*\},\\
		\mathcal{Z}_2[k]&=\{i\in \mathcal{V}: x_i[k]\leq \underline{x}^*\}.
	\end{aligned}
\end{equation*}
In the following, we will show that with positive probability,
\begin{align}\label{shrinking}
	|\big(\mathcal{Z}_1[k] \cup  \mathcal{Z}_2[k] \big)\cap \mathcal{N}|  >  |\big(\mathcal{Z}_1[k+1]  \cup  \mathcal{Z}_2[k+1] \big)\cap \mathcal{N}|.
\end{align}
This will lead us to the desired contradiction.

We first show that with positive probability, 
\begin{align}\label{nodegetout}
	&\{i \in \big(\mathcal{Z}_1[k] \cup  \mathcal{Z}_2[k] \big)\cap \mathcal{N}: \nonumber \\
	&\medspace\medspace\medspace\medspace\medspace\medspace\medspace\medspace
	i \notin \big(\mathcal{Z}_1[k+1] \cup  \mathcal{Z}_2[k+1] \big)\cap \mathcal{N} \} \neq \emptyset.
\end{align}
Clearly, $\mathcal{Z}_1[k]$ and $\mathcal{Z}_2[k]$ are nonempty and disjoint by assumption.
Since the network is $(f + 1, f + 1)$-robust with $l$ hops, then for every pair of nonempty disjoint subsets $\mathcal{Z}_1[k],\mathcal{Z}_2[k]$, at least one of the following conditions holds:
\begin{enumerate}
	\item $\mathcal{X}_{\mathcal{Z}_1[k]}^{f+1}=\mathcal{Z}_1[k]$;
	\item $\mathcal{X}_{\mathcal{Z}_2[k]}^{f+1}=\mathcal{Z}_2[k]$;
	\item  $| \mathcal{X}_{\mathcal{Z}_1[k]}^{f+1}| +| \mathcal{X}_{\mathcal{Z}_2[k]}^{f+1}| \geq f+1$.
\end{enumerate}
Moreover, $| \mathcal{A}| \leq f$.
	Thus, there always exists a normal node $i$ either in $\mathcal{Z}_1[k]\cap \mathcal{N}$ or $\mathcal{Z}_2[k]\cap \mathcal{N}$ such that
\begin{equation}\label{f1paths}
	| \mathcal{I}_{i, \mathcal{Z}_1[k]}^{\mathcal{A}} | \geq f+1, \medspace \textup{or} \medspace\medspace | \mathcal{I}_{i, \mathcal{Z}_2[k]}^{\mathcal{A}} | \geq f+1.
\end{equation}
Without loss of generality, we suppose that node $i \in \mathcal{Z}_1[k]\cap \mathcal{N}$ has this property. 
By assumption, $\forall i \in \mathcal{Z}_1[k]\cap \mathcal{N} $, we have
\begin{align*}
	x_i[k]=\overline{x}^* \medspace, \medspace \forall  k\geq k_c.
\end{align*}
Moreover, node $i$ will use at least one value smaller than $\overline{x}^*$ to update its own value. 
This is because among the $f +1$ independent paths originating from the nodes outside $\mathcal{Z}_1[k]$, all the source nodes have values smaller than $\overline{x}^*$ and node $i$ can remove only values along $f$ independent paths.

Also notice that all the values larger than $\overline{x}^*$ must come from adversary nodes and are removed by step 2 of Algorithm~1 since their MMC has a cardinality of at most $f$.
Consequently, node $i$ updates its value as 
\begin{align}\label{getout}
	x_i[k+1] &\leq Q\big((1-\alpha)\overline{x}^*+ \alpha(\overline{x}^*-1) \big) \nonumber \\
	&=Q\left(\overline{x}^*-\alpha \right).
\end{align}
By \eqref{quantizer}, the quantizer takes the floor function as
$Q\left(\overline{x}^*-\alpha \right)=\overline{x}^*-1$ with probability $\alpha$. Thus, with positive probability, we have
\begin{align*}
	x_i[k+1] \leq \overline{x}^*-1.
\end{align*}
This indicates that with positive probability, $ i \notin \mathcal{Z}_1[k+1] \cap \mathcal{N}$. 
Similarly, if node $i \in \mathcal{Z}_2[k]\cap \mathcal{N}$, then with positive probability, its quantizer chooses the ceiling function, then $ i \notin \mathcal{Z}_2[k+1] \cap \mathcal{N}$. Thus, we have proved \eqref{nodegetout}.

Next, we show that with positive probability,
\begin{align}\label{notgetin}
	(\mathcal{N} \setminus \mathcal{Z}_1[k])  &\cap \mathcal{Z}_1[k+1] = \emptyset, \nonumber\\
	(\mathcal{N} \setminus \mathcal{Z}_2[k])  &\cap \mathcal{Z}_2[k+1] = \emptyset. 
\end{align}
We have $\forall i \in \mathcal{N} \setminus \mathcal{Z}_1[k]$, 
\begin{align*}
	x_i[k] \leq \overline{x}^*-1 \medspace, \medspace \forall  k\geq k_c.
\end{align*}
According to Algorithm~1, all values larger than $\overline{x}^*$ will be discarded by node $i$, and the inequality \eqref{getout} holds for node $i$ too. Therefore, with positive probability $\alpha$, $\{ \mathcal{N} \setminus \mathcal{Z}_1[k] \} \cap \mathcal{Z}_1[k+1] = \emptyset$. Similarly, we can prove that with positive probability, $ (\mathcal{N} \setminus \mathcal{Z}_2[k])  \cap \mathcal{Z}_2[k+1] = \emptyset$.

Combining \eqref{nodegetout} and \eqref{notgetin}, we have proved \eqref{shrinking}. Therefore, for any $k \geq k_c + n_N $, it holds with positive probability that 
\begin{align*}
	\big(\mathcal{Z}_1[k] \cup  \mathcal{Z}_2[k] \big)\cap \mathcal{N}= \emptyset,
\end{align*}
since $|\mathcal{N}|=n_N $. This contradicts that the final values of $\overline{x}[k]$ and $\underline{x}[k]$ are $\overline{x}^*$ and $\underline{x}^*$, respectively, proving (C2).

Lastly, we show (C3). Assume that the normal nodes have
reached the common value $x^*$ at time $k$. Since $| \mathcal{A}| \leq f$, according to Algorithm~1, any node $j \in \mathcal{A}$ with value $x_j [k] \neq x^*$ is ignored by the normal nodes. Thus, $x_i [k + 1] = x^*, \forall i\in \mathcal{N}$. This completes the proof.
\hfill  $\blacksquare$

In Theorem~\ref{syn}, our graph condition is consistent with the one for deterministic real-valued consensus \cite{yuan2021resilient}. However, this note studies agents taking quantized values and the convergence to an exact consensus value can be obtained in a probabilistic sense.
Moreover, with longer relay hops (i.e., by increasing $l$), our algorithm has a faster convergence speed than that of one-hop algorithms, which is also reported in \cite{yuan2021resilient}. 
On the other hand, randomization in our quantizer \eqref{quantizer} is crucial for Algorithm~1. For example, if node $i$ is in the set $\mathcal{Z}_1[k]\cap \mathcal{N}$ with $| \mathcal{I}_{i, \mathcal{Z}_1[k]}^{\mathcal{A}} | \geq f+1$ in \eqref{f1paths}, and if it always takes a deterministic ceiling function in \eqref{getout}, then we cannot guarantee that $ i \notin \mathcal{Z}_1[k+1] \cap \mathcal{N}$; see also the related discussions in \cite{dibaji2018resilient}.

\begin{remark}
	We emphasize that our approach can be applied to the problem of binary consensus as well, which has been studied in \cite{Lynch,fischer1985impossibility,wang2020asynchronous,khan2020exact}. As long as the initial states of all agents are restricted to $0$ and $1$, the safety interval in \eqref{safety1} indicates that the normal agents’ values will remain binary and come to agreement eventually. It should be noted that all results in this note remain true for binary consensus.
\end{remark}

Next, we compare our result in Theorem~\ref{syn} with the one in \cite{khan2020exact}.
There, the authors studied the synchronous binary consensus under the $f$-total malicious model and they provided a necessary and sufficient graph condition for a flooding algorithm to succeed. As we proved in \cite{yuan2021resilient}, our condition is equivalent to the one in \cite{khan2020exact} for directed graphs when our relay range is unbounded.

Here, we explain how their algorithm is different from ours.
Their algorithm is executed in different phases, where each phase corresponds to a possible set of faulty nodes. In each phase (i.e., a phase for a fixed $f$-total set $\mathcal{F}$), a certified propagation algorithm (CPA)\footnote{In the CPA method, when a normal agent receives an identical value from at least $f+1$ neighbors, it commits to such a value.} type of method is used to broadcast the value of a source component to the entire network. In each phase, when the set $\mathcal{F}$ is not the true adversary set, the normal agents will take values equal to some normal agent. When the set $\mathcal{F}$ is the true adversary set, the normal agents will agree on a common value, and hence, binary consensus is achieved. Thus, in the worst case, their algorithm needs to be executed for $\binom{n}{f}$ phases to get a true faulty set and then ensures binary consensus. However, our algorithm does not guess the true faulty set and can guarantee resilient quantized consensus with probability 1. Besides, our algorithm can handle asynchronous updates with delays (in Section~\ref{secasyn}), which cannot be solved by the deterministic algorithm in \cite{khan2020exact}.

\section{Asynchronous Network}\label{secasyn}

	We analyze the asynchronous versions of the QMW-MSR algorithm presented in Algorithm~2 under the malicious and Byzantine models.
If node $i$ does not receive any value along some path $P$ originating from its $l$-hop neighbor $j$, then it takes this value as one malicious empty value.

\begin{algorithm}[t]
	\caption{Asynchronous QMW-MSR Algorithm   }
	\LinesNumbered 
	\KwIn{Node $i$ knows $x_i[0]$, $\mathcal{N}_i^{l-}$, $\mathcal{N}_i^{l+}$. }
	\For{$k\geq0$}{
		\textbf{Decide:} to make an update ($ \delta=1$) or not ($ \delta=0$) independently.
		
		\eIf{   $ \delta=1$ }{
			$\mathcal{M}_i[k] \leftarrow $ the most recently received $x_j^P[k-\tau_{ij}^P[k]]$, $j\in \mathcal{N}_i^{l-}$ on each $l$-hop path.
			
			Steps 2) and 3) of Algorithm~1. 
			
			Send $m_{ij}[k+1]=(x_i[k+1],P_{ij}[k+1])$ to $\forall j\in \mathcal{N}_i^{l+}$. 
		}
		{
			$x_i[k+1]=x_i[k]$.
		}
		\KwOut{$x_i[k+1]$.}
	}
\end{algorithm}

We introduce the following two mechanisms for normal node $i$ to decide whether it makes an update or not.

(i) Under deterministic updates, each normal node
$i$ makes an update at least once in $\overline{k}$ time steps, that is,
\begin{equation}
	\bigcup_{m=k}^{k+\overline{k}-1} \mathcal{U}[m]=\mathcal{N}, \medspace \forall k\in \mathbb{Z}_+,
\end{equation}
while adversary nodes may deviate from this update setting.

(ii) Under randomized updates, each normal
node $i$ makes an update at time $k\geq0$ with probability $p_i\in
(0, 1]$ in an i.i.d. fashion. That is, for node $i$, at each time $k$,
\begin{equation}
	\textup{Prob}\{i \in \mathcal{U}[k]\} =p_i, \medspace \textup{Prob}\{i \notin \mathcal{U}[k]\} =1-p_i.
\end{equation}
Here, the algorithm remains fully distributed since the probabilities $p_i$ at each node can be different.

\subsection{Asynchronous Updates without Delays}\label{secasynnodelay}

	We first analyze Algorithm~2 under randomized updates without delays under the malicious model. An advantage of randomized updates is that the malicious nodes cannot predict the update times of the normal nodes.
Moreover, there is always nonzero probability that all normal nodes in the system update their states simultaneously at each time $k$.
As a result, the necessary and sufficient condition for this case is the same as that for the synchronous case in Theorem \ref{syn}. 

\begin{theorem}\label{asynrandom}
	Consider a directed network $\mathcal{G} = (\mathcal{V},\mathcal{E})$ with $l$-hop communication, where each normal node updates its value according to the asynchronous QMW-MSR algorithm with parameter $f$ under randomized updates without delays. Under the $f$-total malicious model, resilient quantized consensus is achieved almost surely with safety interval \eqref{safety1} if and only if $\mathcal{G}$ is $(f + 1, f + 1)$-robust with $l$ hops.
\end{theorem}

\begin{proof}
	(Necessity) The necessity part is the same as that in the proof of Theorem~\ref{syn}.
	
	(Sufficiency) We show that the conditions (C1)–(C3) in Lemma \ref{synlemma} hold. It is easy to see that conditions (C1) and (C3) hold under the randomized updates too. Therefore, we only need to show (C2). 
	
	We know from (C1) that $\overline{x}[k]$ and $\underline{x}[k]$ will reach their final values $\overline{x}^*$ and $\underline{x}^*$ at some time $k_c$, respectively. 
	Like the proof of Theorem \ref{syn}, we prove (C2) by contradiction. Assume  $\overline{x}^*>\underline{x}^*$. Then the sets $\mathcal{Z}_1[k]$ and $\mathcal{Z}_2[k]$ are disjoint and nonempty. 
	
	We first prove that \eqref{nodegetout} holds with positive probability $\forall k\geq k_c $. Note that under randomized updates, the probability for any normal node to update its value at time $k$ is positive.
	Besides, since $\mathcal{G}$ is $(f + 1, f + 1)$-robustness with $l$ hops and $| \mathcal{A}| \leq f$, there is a normal node $i$ either in $\mathcal{Z}_1[k]\cap \mathcal{N}$ or $\mathcal{Z}_2[k]\cap \mathcal{N}$ such that \eqref{f1paths} holds. Thus, by the same reasoning as in the proof of Theorem \ref{syn}, \eqref{nodegetout} can be proved.
	
	Then we show that with positive probability, \eqref{notgetin} holds. For any node $i\in \mathcal{N}\setminus\mathcal{Z}_1[k]$, with positive probability, it makes an update at time $k$ and will not enter $\mathcal{Z}_1[k+1]$ by the same reasoning as in the proof of Theorem \ref{syn}. Similarly, for any node $i\in \mathcal{N}\setminus\mathcal{Z}_2[k]$, with positive probability, it will not enter $\mathcal{Z}_2[k+1]$. Thus, with \eqref{nodegetout} and \eqref{notgetin}, \eqref{shrinking} can be obtained.
	
	Lastly, we conclude that for any $k \geq k_c + n_N $, one of the two sets $\mathcal{Z}_1[k]\cap \mathcal{N}$ and $\mathcal{Z}_2[k]\cap \mathcal{N}$ is empty with positive probability. We have the desired contradiction.
\end{proof}

\subsection{Asynchronous Updates with Delays}\label{secasyndelay}

Next, we see how the QMW-MSR algorithm performs with delays under the malicious model. This asynchrony setting is also studied in many existing works \cite{dibaji2018resilient,qin2012sufficient,senejohnny2019resilience}.

We employ the quantized control input taking account of possible delays in the values from the multi-hop neighbors as 
\begin{equation}
	u_i[k]=Q\bigg(\sum_{j\in \mathcal{N}_i^{l-}} a_{ij}[k]x_j^P[k-\tau_{ij}^P[k]]\bigg),  
\end{equation}
where $\tau_{ij}^P[k]\in \mathbb{Z}_+$ denotes the delay in the $(j,i)$-path $P$ at time $k$ and $x_j^P[k]$ denotes the value of node $j$ at time $k$ sent along path $P$.
The delays are time-varying and may be different at each path, but we assume the common upper bound $\tau$ on any normal path $P$ (all nodes on path $P$ are normal) as
\begin{equation}
	0\leq \tau_{ij}^P[k] \leq \tau,\medspace j\in \mathcal{N}_i^{l-}, \medspace k\in \mathbb{Z}_+.
\end{equation}
Hence, each normal node $i$ is aware of the value
of each of its normal $l$-hop neighbor $j$ on each normal $(j,i)$-path $P$ at least once in $\tau$ time steps, but
possibly at different time instants \cite{dibaji2018resilient}. 
This assumption also indicates that for each normal node, the gap between two consecutive updates should be less than $\tau$, i.e., $\overline{k}\leq\tau$. 
Besides, agents need not know this bound.
 
Let $D[k]$ be a diagonal matrix whose $i$th entry is
given by $d_i[k]=\sum_{j=1}^{n} a_{ij}[k].$ Then,
let the matrices $A_\gamma[k]\in \mathbb{R}^{n\times n}$ for $ 0\leq \gamma \leq \tau$, and $L_{\tau}[k]\in \mathbb{R}^{n\times (\tau +1)n}$ be given by
\begin{equation}
	A_\gamma[k]=\left\{
	\begin{array}{lll} 
		a_{ij}[k] &\textup{if} \thinspace i\neq j \thinspace\textup{and}\thinspace \tau_{ij}[k]=\gamma,\\
		0 & \textup{otherwise,}
	\end{array}
	\right.
\end{equation}
and $L_{\tau}[k]=\Big[ D[k]-A_0[k] \medspace\medspace\medspace -A_1[k] \medspace\medspace\medspace \cdots \medspace\medspace\medspace -A_{\tau}[k] \Big].$
Note that the summation of the entries of each row of $L_{\tau}[k]$ is zero.

Now, the control input can be expressed as
\begin{equation}
	\begin{array}{lll} 
		u^N[k] =Q\big(-L_{\tau}^N[k]z[k]\big),\\
		u^A[k] : \textup{arbitrary,}
	\end{array}
\end{equation}
where $z[k]= [x[k]^T x[k-1]^T \cdots  x[k-\tau]^T]^T$ is a $(\tau+1)n$-dimensional vector for $k\geq0$ and $L_{\tau}^N[k]$ is a matrix formed by the first $n_N$ rows of $L_{\tau}[k]$. Here, $z[0] = [x[0]^T 0^T \cdots 0^T]^T$. Then, the agent dynamics can be written as
\begin{equation} \label{system2}
	x[k+1]=Q\big(\Gamma[k] z[k]\big) +   \begin{bmatrix} 0  \\ I_{n_A} \end{bmatrix}u^A[k],
\end{equation}
where $\Gamma[k]$ is an  $n\times(\tau+1)n$ matrix given by $\Gamma[k] = \begin{bmatrix} I_n &  0 \end{bmatrix} -   \begin{bmatrix} L_{\tau}^N[k]^T  & 0 \end{bmatrix}^T. $
The safety interval is the same as the synchronous one and can be rewritten as
\begin{equation}  \label{safety2}
	\mathcal{S}_{\tau}=\Big[ \min z^N[0], \max z^N[0] \Big].            	
\end{equation}

The main result of this section now follows. Its proof can be found in the Appendix~C.

\begin{theorem}\label{asyndetermin}
	Consider a directed network $\mathcal{G} = (\mathcal{V},\mathcal{E})$ with $l$-hop communication, where each normal node updates its value according to the asynchronous QMW-MSR algorithm with parameter $f$ under deterministic updates with delays. Under the $f$-total/local malicious model, resilient quantized consensus is achieved almost surely only if $\mathcal{G}$ is $(f + 1, f + 1)$-robust with $l$ hops.
	Moreover, if $\mathcal{G}$ is $(f + 1)$-strictly robust with $l$ hops, then resilient quantized consensus is reached almost surely with safety interval \eqref{safety2}.
\end{theorem}

The work \cite{dibaji2018resilient} provided a sufficient condition for resilient quantized consensus for the one-hop case under asynchronous deterministic updates with delays, which is that the graph $\mathcal{G}$
is $(2f + 1)$-robust. By Lemma~\ref{abc}, if graph $\mathcal{G}$ is $(2f + 1)$-robust, then it is $(f + 1)$-strictly robust (with $1$ hop), but not vice versa. Hence, our sufficient condition is tighter than that in \cite{dibaji2018resilient}. However, checking strict robustness of a given graph is more complex compared to checking robustness since we need to check $\binom{n}{f}$ times for robustness of $\mathcal{G}_{\mathcal{H}}$.

\begin{remark}
In Theorem~\ref{asynrandom}, through randomized updates, we are able to derive a tighter sufficient condition for quantized consensus of asynchronous networks without delay. However, when communication delays are present, that condition is not sufficient anymore even if we apply randomized updates. The reason is that if nonuniform delays are in presence in different paths, then malicious nodes can make updates frequently, and their delayed values received by the normal neighbors may appear different for different neighbors. In this case, malicious nodes can have the attack ability close to Byzantine nodes. Hence, the stricter sufficient condition in Theorem~\ref{asyndetermin} is needed.
\end{remark}

\subsection{Quantized Consensus under the Byzantine Model}\label{secbyzan}

To handle the Byzantine model, the necessary condition needs to be even stricter than the one for the malicious model. Moreover, the necessary and sufficient conditions for the synchronous and asynchronous QMW-MSR algorithms remain the same. Recall that the (strict) robustness notion is different under the $f$-total and $f$-local models by Definition~\ref{rs-robust}.

\begin{theorem}\label{theorem_byzan}
	Consider a directed network $\mathcal{G} = (\mathcal{V},\mathcal{E})$ with $l$-hop communication, where all normal nodes update values according to the synchronous QMW-MSR algorithm or the asynchronous QMW-MSR algorithm under randomized/deterministic updates with/without delays. Under the $f$-total/local Byzantine model, resilient quantized consensus is achieved almost surely with safety interval \eqref{safety2} if and only if $\mathcal{G}$ is $(f + 1)$-strictly robust with $l$ hops.
\end{theorem}

The proof can be found in the Appendix~D.
Recently, \cite{wang2020asynchronous} studied asynchronous EBC in undirected networks, and they used a flooding algorithm similar to that in \cite{khan2020exact}, but with a randomization process at each node. They proved that the necessary and sufficient graph condition for asynchronous EBC is the same as the one for the synchronous case \cite{dolev1982byzantine}: (i) $n>3f$ and (ii) node-connectivity is no less than $2f+1$. This fact is also observed in our results: Our graph condition with $l^*$-hop communication in undirected networks is the same as the above condition; this was proved in \cite{yuan2022asynchronous}.

\begin{remark}
		We compare our result in Theorem~\ref{theorem_byzan} with the one from \cite{wang2020asynchronous}.
		First, their algorithm may take innumerous time for receiving the $n-f$ values in each round due to majority voting on each node. In contrast, our algorithm does not wait until each node receives enough values and it would continue to update its value towards consensus at each local time step.
		We emphasize that \cite{wang2020asynchronous} and \cite{khan2020exact} can handle only the case of binary consensus, while our method can achieve integer-valued consensus in directed networks.
		Moreover, we have seen in Lemma~\ref{lemmastarcyle} that for some graphs, we can achieve the same level of resilience with a much smaller number of hops compared with the flooding algorithm \cite{wang2020asynchronous}.  
\end{remark}

To support the above comparison, we briefly explain how the algorithm in \cite{wang2020asynchronous} works.
There, EBC is achieved with high probability and the majority voting algorithm plays a key role. Their algorithm is executed in phases, but such a notion of phases is close to that of time slots and does not depend on a given faulty node set.
In particular, they used an authenticated double-echo broadcast algorithm, where identical messages of any normal agent will be delivered to all nodes. In each phase, a normal node $i$ has different executions in three rounds. In the first two rounds, node $i$ will determine the majority value among those received from others. If its own value is the same as the majority value, then it broadcasts an empty value in the third round. Otherwise, it sets its value to the majority value and broadcasts the new value in the third round. Then, it collects the values from others again. If more than $f$ values are the same as its own, then it keeps its value for the next phase.
If there is no value existing in more than $f$ messages, it makes a coin toss to get a random value and keeps this value in the next phase. Finally, node $i$ terminates when it finds that $n-f$ values are the same as its own.

\section{Numerical Examples}\label{sec_sim}

We conduct simulations of the QMW-MSR algorithm.
In all figures of time responses in this section, the red dashed lines represent the values from/through adversaries, and the lines not in red color represent the values of normal nodes.

\subsection{Synchronous QMW-MSR Algorithm}\label{sec_4hops}

First, we conduct simulations for the synchronous QMW-MSR algorithm.
Consider the undirected network in Fig.~\ref{graph1}(a) with the adversary set $\mathcal{A}=\{8\}$. 
Let the initial states be $x[0]=[4\ 5\ 6\ 7\ 8\ 9\ 3\ 1]^T$. 
As discussed in Section~\ref{sec_robustness}, this graph is not $(2, 2)$-robust (with $1$ hop), and hence, is not robust enough to tolerate the $1$-total malicious model using the one-hop QW-MSR algorithm from \cite{dibaji2018resilient}. Besides, it is not $2$-strictly robust (with $1$ hop) and is not robust enough to tolerate the $1$-total Byzantine model \cite{vaidya2012iterative}.
Moreover, by Lemma~\ref{lemmacycle}, it is not $(2, 2)$-robust with $l\leq 3$ hops but is $(2, 2)$-robust with $4$ hops.

\begin{figure}[t]
	\centering
	\subfigure[\scriptsize{One-hop algorithm.}]{
		\includegraphics[width=3.3in,height=1.4in]{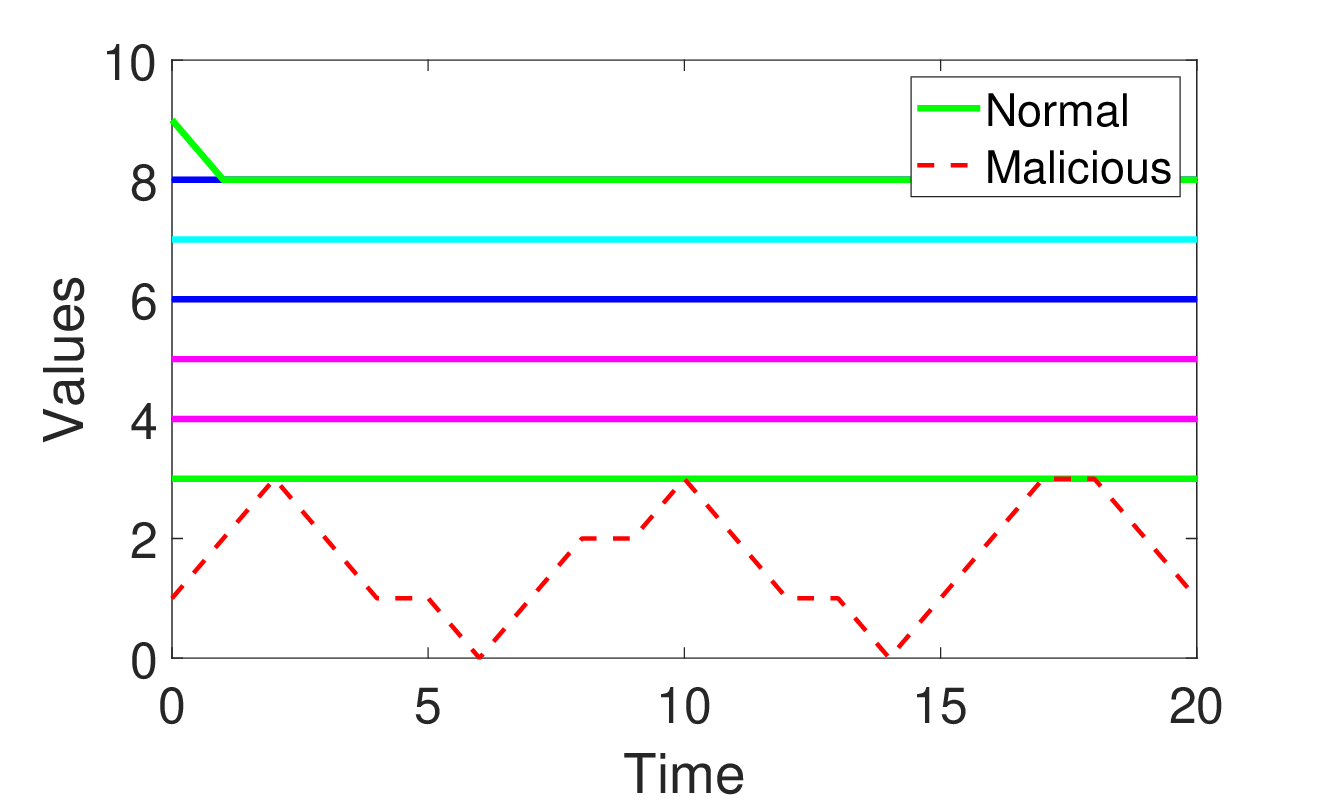}
	}
	\vspace{-4pt}
	
	\subfigure[\scriptsize{Four-hop algorithm.}]{
		\includegraphics[width=3.3in,height=1.4in]{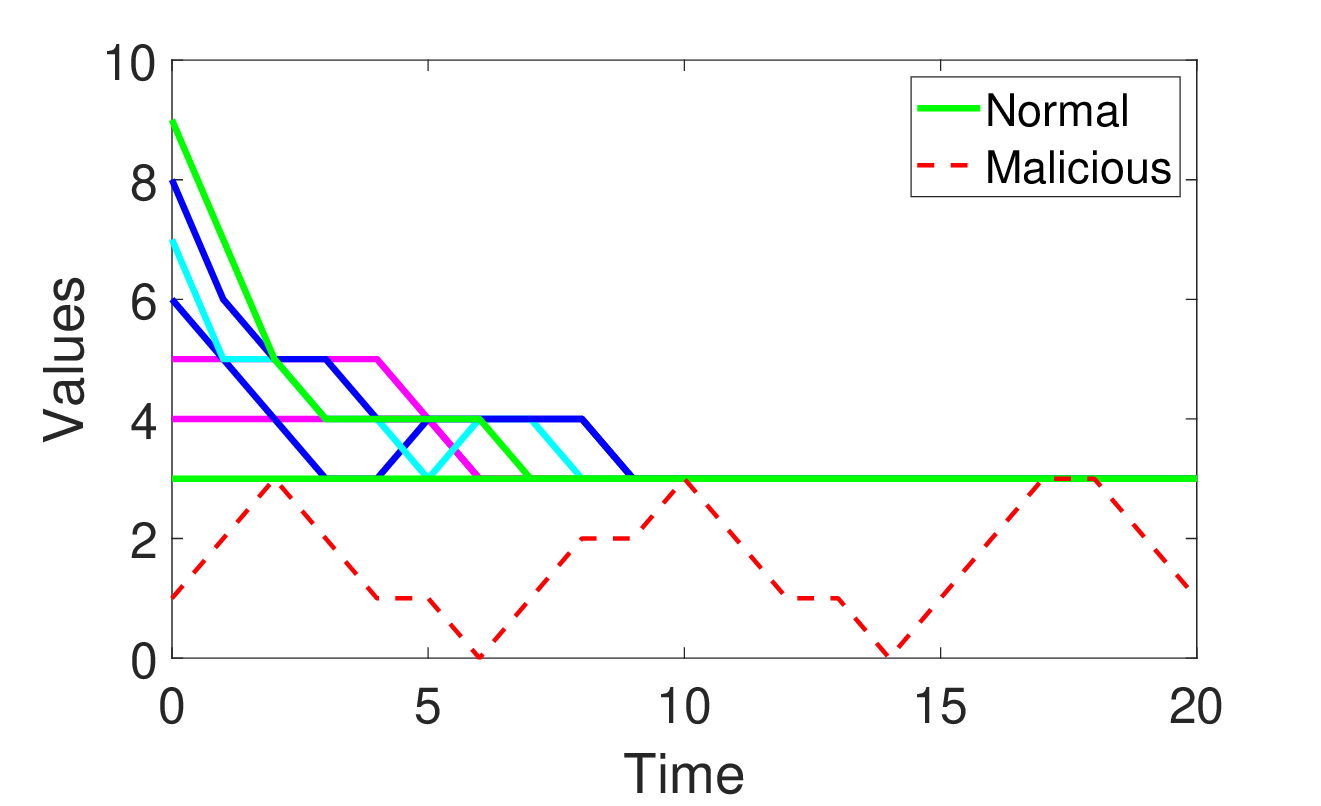}
	}
	\vspace{-4pt}
	\caption{Time responses of the synchronous QMW-MSR algorithm under the 1-total malicious model in the network of Fig.~\ref{graph1}(a).}
	\label{8-syn-malicious}
	\vspace*{-3.0mm}
\end{figure}

\begin{figure}[t]
	\centering
	\includegraphics[width=3.3in,height=1.4in]{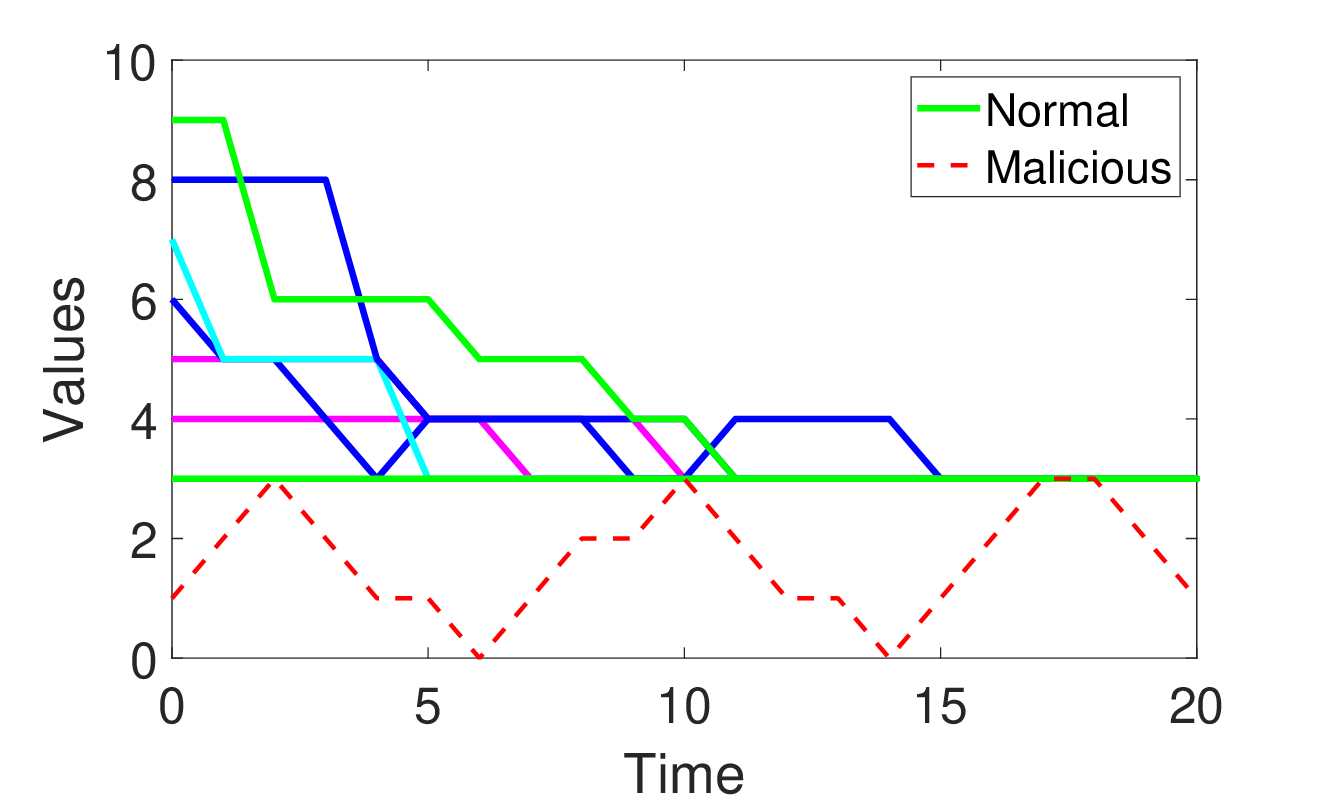}
	\vspace{-4pt}
	\caption{Time responses of the asynchronous four-hop QMW-MSR algorithm under the 1-total malicious model without delays in the network of Fig.~\ref{graph1}(a).}
	\label{8-asyn-random}
	\vspace*{-3.0mm}
\end{figure}

\textit{1) Malicious model: }
We set node 8 to be malicious and let its value evolve based on a quantized sine function w.r.t. time. Suppose that node 8 changes all the relayed values to the same value as its own at each time.
Then, normal nodes update their values using the one-hop, two-hop, three-hop, and four-hop versions of the QMW-MSR algorithm. (Note that the one-hop QMW-MSR algorithm is equivalent to the QW-MSR algorithm from \cite{dibaji2018resilient}.)
Time responses of the one-hop and four-hop algorithms are given in Figs.~\ref{8-syn-malicious}(a) and \ref{8-syn-malicious}(b), respectively.
Observe that consensus is not achieved until normal nodes use the four-hop algorithm, which verifies the result in Theorem~\ref{syn}.

\textit{2) Byzantine model: }
The cycle network is not $2$-strictly robust with any number of hops since the minimum in-degree of an $(f+1)$-strictly robust graph is $2f+1$ (see \cite{su2017reaching,yuan2022asynchronous}). Therefore, the cycle network cannot achieve resilient quantized consensus under the Byzantine model unless increasing edges or utilizing external authentication. This example has revealed the gap between the graph conditions for the malicious and Byzantine models.

\subsection{Asynchronous QMW-MSR Algorithm}

\textit{1) Without delays: }
First, we consider asynchronous randomized updates without delays for the cycle network in Fig.~\ref{graph1}(a). For malicious node 8, we assume that it modifies its own and relayed values.  
Yet, observe in Fig.~\ref{8-asyn-random} that quantized consensus is achieved using the four-hop algorithm. This verifies the result in Theorem~\ref{asynrandom}.

\begin{figure}[t]
	\centering
	\includegraphics[width=3.3in,height=1.4in]{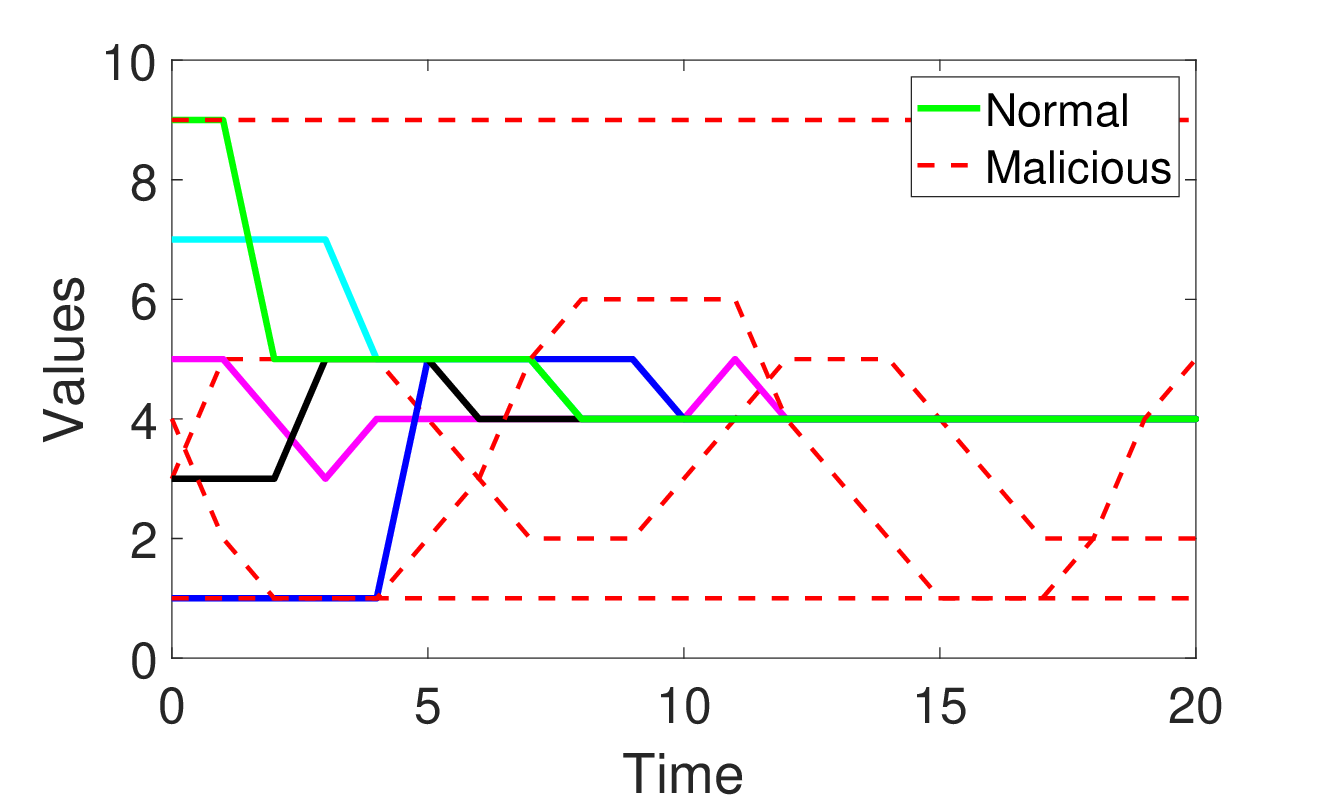}
	\vspace{-4pt}
	\caption{Time responses of the asynchronous two-hop QMW-MSR algorithm under the 1-total malicious model with delays in the network of Fig.~\ref{graph1}(b).}
	\label{6-two-hop-asyn-delay}
	\vspace*{-3.0mm}
\end{figure}

\begin{figure}[t]
	\centering
	\includegraphics[width=1.4in]{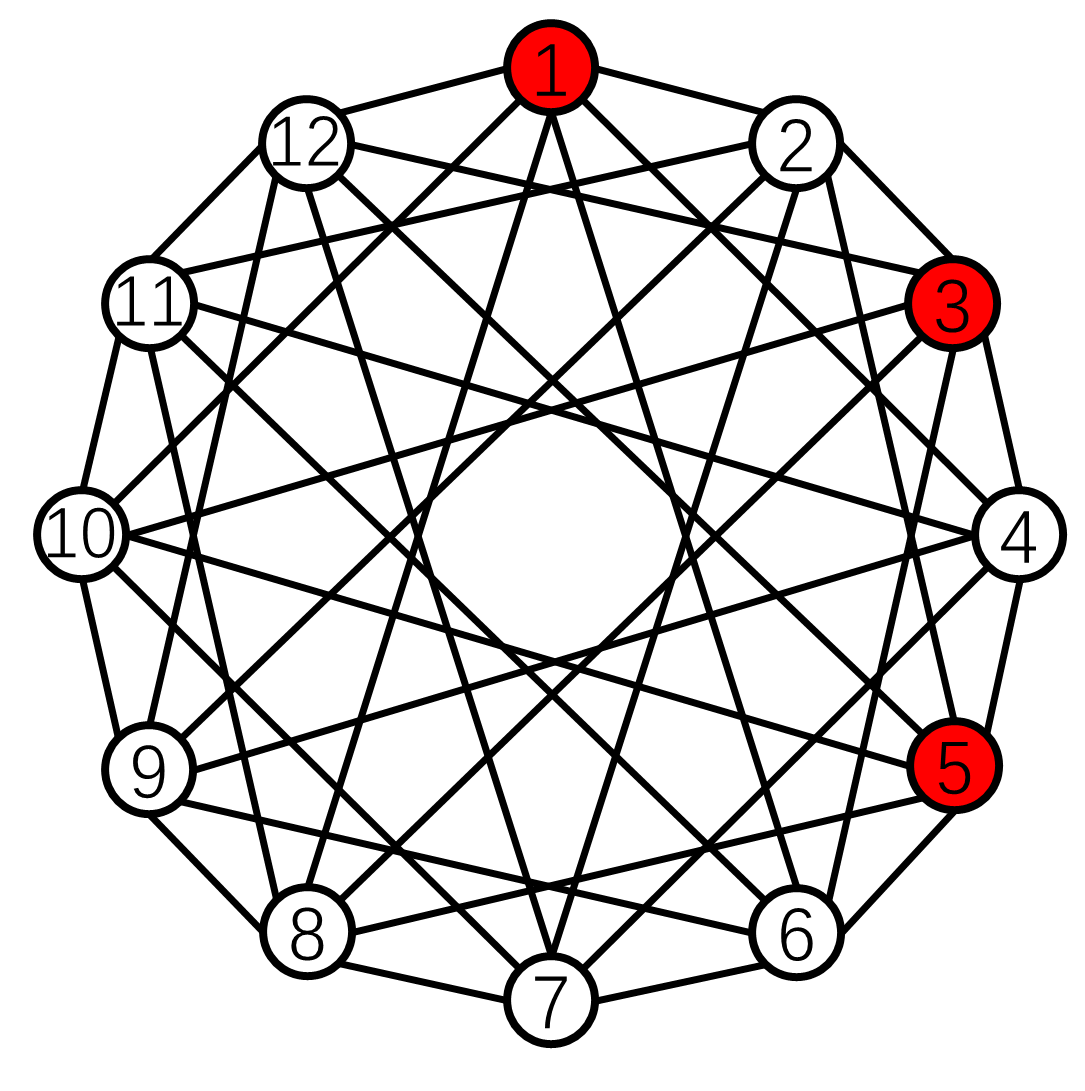}
	\vspace{-4pt}
	\caption{The undirected graph is (3, 3)-robust with 1 hop and is (4,4)-robust with 2 hops.}
	\label{12node_graph}
	\vspace*{-3.0mm}
\end{figure}

\textit{2) With delays: }
Consider the network in Fig.~\ref{graph1}(b) with $\mathcal{A}=\{1\}$. Let $x[0]=[3\ 5\ 1\ 7\ 3\ 9]^T$. With the one-hop algorithm \cite{dibaji2018resilient}, it is not $2$-robust and hence is not robust enough to tolerate $1$-total malicious/Byzantine model. Next, consider the two-hop algorithm under asynchronous deterministic updates with delays for this network. It becomes $2$-strictly robust with $2$ hops.
For malicious node 1, we assume that it not only manipulates its own value but also relays different values to neighbors asynchronously (i.e., it sends out four different values). Moreover, the delays for the messages from one-hop and two-hop neighbors are set as 0 and 1, respectively.
Observe in Fig.~\ref{6-two-hop-asyn-delay} that resilient quantized consensus is attained.

As we discussed earlier, malicious agents under communication delays can cause similar attacks as Byzantine agents. Hence, the simulation for the Byzantine model is similar to that depicted in Fig.~\ref{6-two-hop-asyn-delay}. Besides, the algorithm in \cite{wang2020asynchronous} can handle asynchronous EBC in only undirected networks and is thus not applicable in this case.

\subsection{Simulations in a Medium-Sized Network with Tolerance Limit of Malicious Agents}\label{sec_sim_12}

Consider the network under the 3-total malicious model in Fig.~\ref{12node_graph}.
From Lemma 7.6 in \cite{yuan2021resilient}, we know that to satisfy $(f+1,f+1)$-robustness with $l$ hops in Theorem~\ref{syn}, a network must have minimum in-degree of $2f$. Thus, this network can tolerate at most 3 malicious agents since its minimum in-degree is 6.

\begin{figure}[t]
	\centering
	\subfigure[\scriptsize{One-hop algorithm.}]{
		\includegraphics[width=3.3in,height=1.4in]{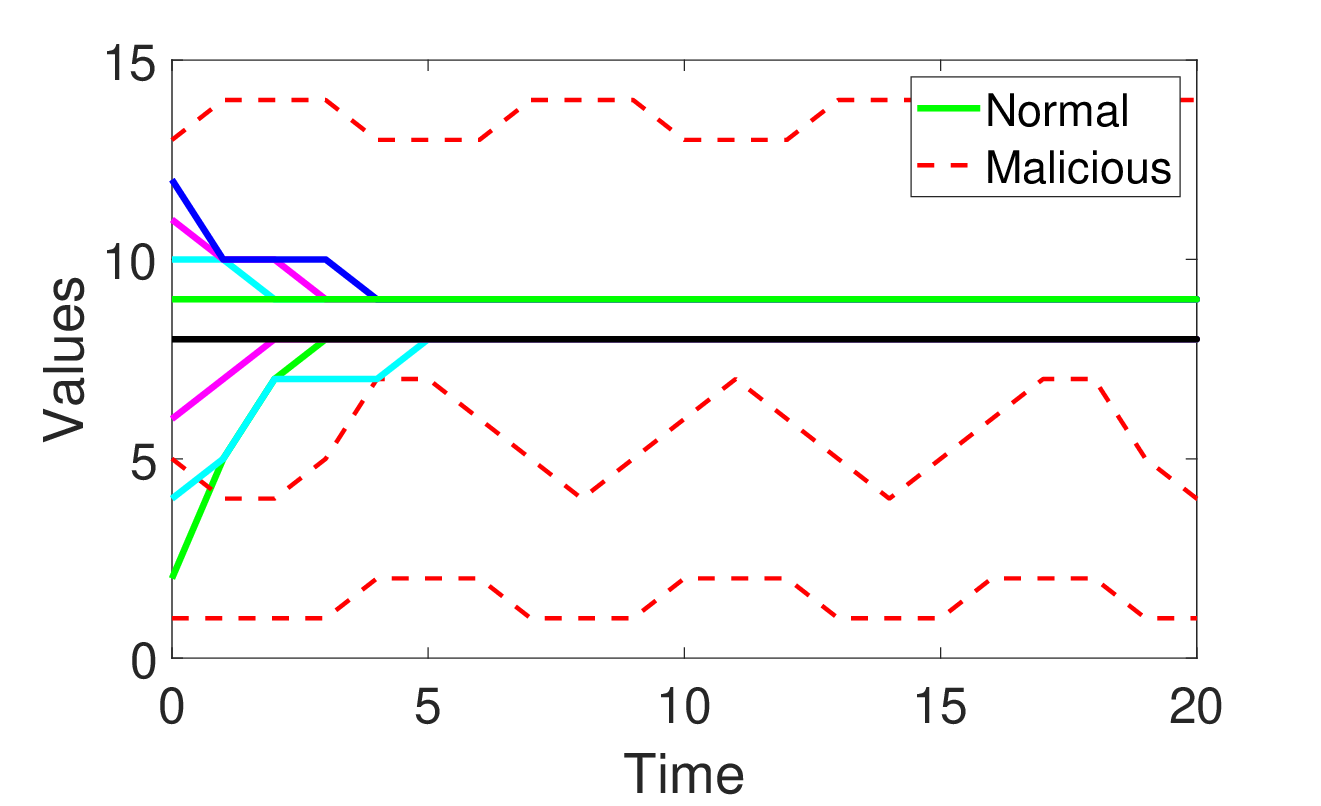}
	}
	\vspace{-4pt}
	
	\subfigure[\scriptsize{Two-hop algorithm.}]{
		\includegraphics[width=3.3in,height=1.4in]{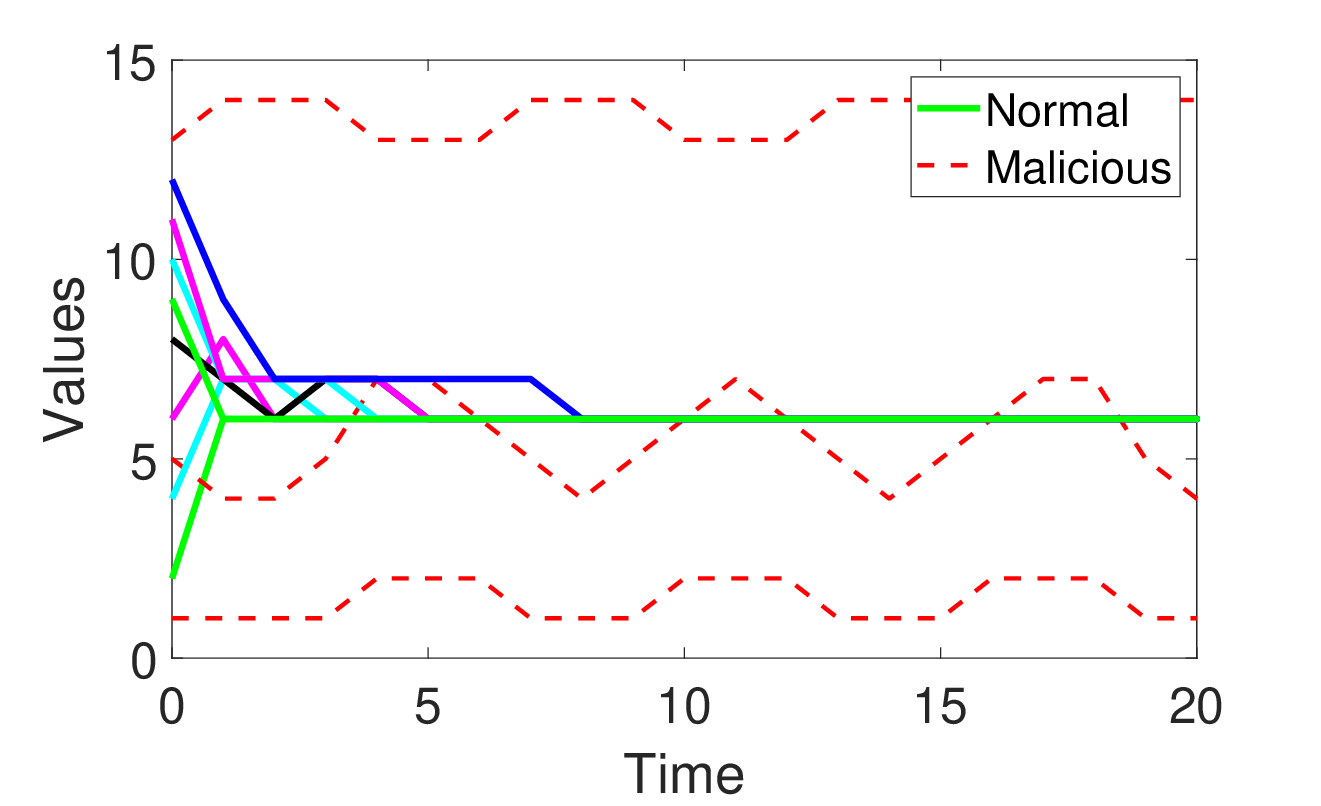}
	}
	\vspace{-4pt}
	\caption{Time responses of the synchronous QMW-MSR algorithm under the 3-total malicious model in the network of Fig.~\ref{12node_graph}.}
	\label{12-mali-135}
	\vspace*{-3.0mm}
\end{figure}

Let the set of adversaries $\mathcal{A}=\{1,3,5\}$.
This network is not $(4, 4)$-robust (with $1$ hop), e.g., consider the node sets $ \{1,2,3,4,5,6\}$ and $ \{7,8,9,10,11,12\}$. Yet, by Lemma~\ref{lemma2f}, it is $(4, 4)$-robust with $2$ hops.
Let $x_i[0]\in (0,15), \forall i \in \mathcal{V}$, and suppose that the malicious nodes change their own and relayed values to oscillating integers as the red dashed lines in Fig.~\ref{12-mali-135}. 
Observe in Fig.~\ref{12-mali-135}(a) that quantized consensus is not achieved using the one-hop algorithm. 
However, in Fig.~\ref{12-mali-135}(b), quantized consensus is achieved using the two-hop algorithm.

In this example, we have achieved the maximum tolerance of the malicious agents with our algorithm using only two-hop communication. In contrast, the flooding algorithm from \cite{khan2020exact} requires each normal node to send its value to all the nodes in the network along at most $(n-1)$-hop paths. Hence, our algorithm is more efficient in exploiting the ability of the multi-hop techniques to achieve the same level of tolerance of malicious agents. Besides, our algorithm can guarantee resilient integer-valued consensus and can be extended to asynchronous updates with delays. These features highlight our advantages over \cite{khan2020exact}.

\section{Conclusion}

We have proposed a multi-hop algorithm to solve resilient quantized consensus with asynchronous updates and delays. 
We have accordingly proved necessary and sufficient conditions for our algorithm to succeed under the malicious and Byzantine models. Our method requires sparser networks compared to the heavy graph requirements for the one-hop algorithms.
Compared to the existing methods studying binary consensus, our algorithm considers the more general $l$-hop case and can handle the binary and integer-valued consensus problems. Lastly, our algorithm is efficient in using less relay hops to achieve the same level of tolerance for adversary agents compared to the flooding algorithms. In future work, we intend to characterize tight graph conditions for our algorithm to succeed in time-varying networks.

\appendices

\section*{Appendix}

\section*{A. Proof of Lemma \ref{lemma2f}}

\begin{proof}
	We first claim that $\mathcal{G}=\mathcal{K}_{n_1,n_2}$ reaches its maximum robustness with $l\geq 2$ hops. To this end, consider any node $i\in \mathcal{V}_1$. It connects with all other nodes in $\mathcal{V}_1$ with two-hop paths and further connects with all nodes in $\mathcal{V}_2$ with three-hop paths. Notice that nodes in $\mathcal{V}_2$ are direct neighbors of node $i\in \mathcal{V}_1$. Hence, by Definition~\ref{rs-robust}, the robustness of $\mathcal{G}$ will not increase for $l\geq 2$ hops.

	Proposition~7.1 in \cite{yuan2021resilient} indicates that our condition ($(f +1 , f +1)$-robust with $l$ hops) is equivalent to the one in \cite{khan2020exact} for the unbounded path length case. Their condition is given by: (i) The minimum in-degree of the graph is $2f$ and (ii) the graph is $\lfloor \frac{3f }{2} \rfloor$-connected. Furthermore, we conclude that $\mathcal{G}$ in Fig.~\ref{example_graphs}(a) is $d$-connected. Thus, if $\mathcal{G}$ satisfies condition (i), then it also meets condition (ii). Therefore, $\mathcal{G}$ is $(\lfloor \frac{d }{2} \rfloor +1 , \lfloor \frac{d }{2} \rfloor +1)$-robust with $l^*$ hops. By our first claim, we conclude that $\mathcal{G}$ is $(\lfloor \frac{d }{2} \rfloor +1 , \lfloor \frac{d }{2} \rfloor +1)$-robust with $l\geq 2$ hops under the $\lfloor \frac{d }{2} \rfloor$-total model.
\end{proof}

\section*{B. Proof of Lemma \ref{lemmastarcyle}}

\begin{proof}
	The longest cycle-free path length in $\mathcal{G}=\mathcal{W}_n$ is $l^*=n-1$. To proceed, consider two cases for the $1$-total set $\mathcal{F}=\{m\}$: (i) When node $m$ is on the cycle and (ii) when node $m$ is the centering node. Recall in Definition~\ref{strict_robustness} that $\mathcal{G}_{\mathcal{H}}$ is the remaining graph after removing the set $\mathcal{F}$ from $\mathcal{G}$.
	
	For case (i), we can easily observe that after removing node $m$ from the cycle, the remaining graph $\mathcal{G}_{\mathcal{H}}$ is $2$-robust (with $1$ hop).

	For case (ii), after removing the centering node $m$, the remaining graph $\mathcal{G}_{\mathcal{H}}$ is a cycle graph. Then, we will prove that $\mathcal{G}_{\mathcal{H}}$ is $2$-robust with $\lfloor l^*/4 \rfloor+1$ hops. Hence, by Definition~\ref{rs-robust} (with $r=2$, $s=1$), we need to prove that for every pair of nonempty, disjoint subsets $\mathcal{V}_\text{1},\mathcal{V}_\text{2}\subset \mathcal{V}$, at least one node in $\mathcal{V}_\text{1} \cup\mathcal{V}_\text{2}$ has 2 independent paths originating from the nodes outside its corresponding set. We conclude that when $\mathcal{V}_\text{1}$ and $\mathcal{V}_\text{2}$ are selected as those in Fig.~\ref{example_graphs}(b), the centering node in $\mathcal{V}_\text{1}$ or $\mathcal{V}_\text{2}$ has 2 shortest paths from the nodes outside the set. In this case, the length of such paths is of $\lfloor l^*/4 \rfloor+1$ hops. Note that for the nodes other than the centering node in $\mathcal{V}_\text{1}$ or $\mathcal{V}_\text{2}$, they have a shorter path and a longer path from the nodes outside the set, compared to the centering node's paths. Thus, they require longer hops for satisfying the robustness notion.

	Therefore, we conclude that the wheel graph $\mathcal{G}$ is $2$-strictly robust with $\lfloor l^*/4 \rfloor+1$ hops under the $1$-total model.
\end{proof}

\section*{C. Proof of Theorem \ref{asyndetermin}}

\begin{proof}
	(Necessity) The synchronous network is a special case of the asynchronous network. Thus, the necessary condition for synchronous case holds for asynchronous case.
	
	(Sufficiency) We show that the conditions (C1)–(C3) in Lemma \ref{synlemma} hold. 
	To show (C1), define the minimum and maximum values of the normal agents at time
	$k$ and the previous $\tau$ time steps by
	\begin{equation}
		\begin{array}{lll} 
			\overline{z}[k] =\max \left( x^N[k], x^N[k-1],\dots, x^N[k-\tau]\right),\\
			\underline{z}[k] = \min \left( x^N[k], x^N[k-1],\dots, x^N[k-\tau]\right).
		\end{array}
	\end{equation}
	Then we prove that $\overline{z}[k]$ is a nonincreasing function.
	By \eqref{system2}, at time $k \geq 1$, each normal agent updates its value based on a quantized convex combination of the neighbors' values from $k$ to $k-\tau$. We know from step 2 of Algorithm~1 that the values outside of the interval $[\underline{z}[k],\overline{z}[k]]$ will be ignored. Hence, we obtain $x_i[k+1] \leq \max \left( x^N[k], x^N[k-1],\dots, x^N[k-\tau]\right)=\overline{z}[k]$, $\forall i \in \mathcal{N}$. It also follows that
	\begin{equation*}
		\begin{aligned}
			x_i[k] &\leq \max \left( x^N[k], x^N[k-1],\dots, x^N[k-\tau]\right),\\
			x_i[k-1] &\leq \max \left( x^N[k], x^N[k-1],\dots, x^N[k-\tau]\right),\\
			&\vdots\\
			x_i[k+1-\tau]&\leq \max \left( x^N[k], x^N[k-1],\dots, x^N[k-\tau]\right),
		\end{aligned}
	\end{equation*}
	$\forall i \in \mathcal{N}$. Therefore, we have
	\begin{equation*}
		\begin{aligned}
			\overline{z}[k+1] &=\max \left( x^N[k+1], x^N[k],\dots, x^N[k+1-\tau]\right)\\
			&\leq \max \left( x^N[k], x^N[k-1],\dots, x^N[k-\tau]\right)=\overline{z}[k].
		\end{aligned}
	\end{equation*}
	We can similarly prove that $\underline{z}[k]$ is nondecreasing in time, i.e., $\underline{z}[k+1]\geq \underline{z}[k]$. This indicates that for $k \geq 1$, we have $x_i[k]\in \mathcal{S}_{\tau}$ for $ i \in \mathcal{N}$. When $k=0$, the safety condition clearly holds. Thus, (C1) holds $\forall k\geq 0$.
	
	Next, we show (C2). Since $\overline{z}[k]$ and $\underline{z}[k]$ are
	contained in $\mathcal{S}_{\tau}$ and are monotone, there is a finite time $k_c$ such
	that they both reach their final values with probability 1, denoted by
	$\overline{z}^*$ and $\underline{z}^*$,
	respectively. We will prove by contradiction that $\overline{z}^*=\underline{z}^*$.
	Assume $\overline{z}^*>\underline{z}^*$. When $k\geq k_c$, we define the sets
	\begin{equation*}
		\begin{aligned}
			\mathcal{Z}_1[k]&=\{i\in \mathcal{N}: x_i[k]= \overline{z}^*\},\\
			\mathcal{Z}_2[k]&=\{i\in \mathcal{N}: x_i[k]= \underline{z}^*\}.
		\end{aligned}
	\end{equation*}
	In the following, we will show that with positive probability,
	\begin{align}\label{shrinking_deterministic}
		|\mathcal{Z}_1[k] \cup  \mathcal{Z}_2[k] |  >  |\mathcal{Z}_1[k+1]  \cup  \mathcal{Z}_2[k+1] |.
	\end{align}
	This will lead us to the desired contradiction.

	Due to the convergence of $\overline{z}[k]$ and $\underline{z}[k]$, we have
	\begin{align}\label{nonempty_deterministic}
		\bigcup_{\gamma=0}^{\tau} \mathcal{Z}_1[k_c+\gamma] \neq \emptyset, \medspace\medspace
		\bigcup_{\gamma=0}^{\tau} \mathcal{Z}_2[k_c+\gamma] \neq \emptyset.
	\end{align}
	We claim that $\mathcal{Z}_1[k_c] \neq \emptyset $. To prove this, it is sufficient to show that if $\mathcal{Z}_1[k_c+\gamma] = \emptyset$, then the probability of $\mathcal{Z}_1[k_c+\gamma+1] = \emptyset$ is nonzero for $\gamma=0,\dots,\tau$, which contradicts \eqref{nonempty_deterministic}.
	
	First, we show that with positive probability,
	\begin{align}\label{notgetin_deterministic}
		(\mathcal{N} \setminus \mathcal{Z}_1[k_c+\gamma])  &\cap \mathcal{Z}_1[k_c+\gamma+1] = \emptyset, \nonumber\\
		(\mathcal{N} \setminus \mathcal{Z}_2[k_c+\gamma])  &\cap \mathcal{Z}_2[k_c+\gamma+1] = \emptyset. 
	\end{align}
	Clearly, if there is no node updating at time $k_c+\gamma$, then
	no node can enter $\mathcal{Z}_1[k_c+\gamma+1]$. However, we know that
	each normal node makes an update at least once in $\overline{k}\leq \tau$ time steps.
	Assume that normal node $i$ makes an update at time $k_c+\gamma$.
	Since $\mathcal{Z}_1[k_c+\gamma] = \emptyset $ by assumption,
	$\forall i \in \mathcal{N} \setminus \mathcal{Z}_1[k]$, we have
	\begin{align*}
		x_i[k_c+\gamma] \leq \overline{z}^*-1 \medspace, \medspace \forall  \gamma.
	\end{align*}
	Also notice that all the values larger than $\overline{z}^*$ come from adversary nodes and are removed by step 2 of Algorithm~1. Then we have
	\begin{align}\label{getout2_deterministic}
		x_i[k_c+\gamma+1] &\leq Q\big((1-\alpha)\overline{z}^*+ \alpha(\overline{z}^*-1) \big)  \nonumber\\
		&=Q\left(\overline{z}^*-\alpha \right),
	\end{align}
	where the quantizer takes $Q\left(\overline{x}^*-\alpha \right)=\overline{x}^*-1$ with probability $\alpha$. Therefore, we have proved \eqref{notgetin_deterministic}. Moreover, it must be $\mathcal{Z}_1[k_c] \neq \emptyset $. Similarly, $\mathcal{Z}_2[k_c] \neq \emptyset $ can also be proved.

	Next, we show that with positive probability, 
	\begin{align}\label{nodegetout_deterministic}
		&\{i \in \big(\mathcal{Z}_1[k_c+\gamma] \cup  \mathcal{Z}_2[k_c+\gamma] \big): \nonumber \\
		&\medspace\medspace\medspace\medspace\medspace\medspace\medspace\medspace
		i \notin \big(\mathcal{Z}_1[k_c+\gamma+1] \cup  \mathcal{Z}_2[k_c+\gamma+1] \big) \} \neq \emptyset.
	\end{align}
	Since $\mathcal{G}$ is $(f + 1)$-strictly robust with $l$ hops, by Definition~\ref{strict_robustness}, the normal network must be $(f + 1)$-robust with $l$ hops. Then, for nonempty and disjoint subsets $\mathcal{Z}_1[k_c+\gamma]$ and $\mathcal{Z}_2[k_c+\gamma]$, at least one of the following conditions holds:
	\begin{enumerate}
		\item $\mathcal{X}_{\mathcal{Z}_1[k_c+\gamma]}^{f+1}=\mathcal{Z}_1[k_c+\gamma]$;
		\item $\mathcal{X}_{\mathcal{Z}_2[k_c+\gamma]}^{f+1}=\mathcal{Z}_2[k_c+\gamma]$;
		\item  $| \mathcal{X}_{\mathcal{Z}_1[k_c+\gamma]}^{f+1}| +| \mathcal{X}_{\mathcal{Z}_2[k_c+\gamma]}^{f+1}| \geq 1$.
	\end{enumerate}
	\noindent Thus, there always exists a normal node $i \in \mathcal{N}$ either in $\mathcal{Z}_1[k_c+\gamma]$ or $\mathcal{Z}_2[k_c+\gamma]$ such that $| \mathcal{I}_{i, \mathcal{Z}_1[k_c+\gamma]}^{\mathcal{A}} | \geq f+1$, or $| \mathcal{I}_{i, \mathcal{Z}_2[k_c+\gamma]}^{\mathcal{A}} | \geq f+1$. Without loss of generality, assume that node $i\in \mathcal{Z}_1[k_c+\gamma]$. Then, we have
	\begin{align*}
		x_i[k_c+\gamma]=\overline{z}^* \medspace, \medspace \forall  \gamma.
	\end{align*}
	Due to the $f$-total/local model of adversary nodes, node $i$ will discard all the values larger than $\overline{z}^*$ by step 2 of Algorithm~1, and it uses at least one value from the normal node $j\in \mathcal{N} \setminus \mathcal{Z}_1[k_c+\gamma]$, where
	\begin{align*}
		x_i[k_c+\gamma]\leq \overline{z}^*-1 \medspace, \medspace \forall  \gamma.
	\end{align*}
	Thus, \eqref{getout2_deterministic} holds for node $i$ too.
	Therefore, with positive probability, if node $i\in \mathcal{Z}_1[k_c+\gamma]$, then it gets out of $\mathcal{Z}_1[k_c+\gamma+1]$. Similarly, we can prove that with positive probability, if node $i\in \mathcal{Z}_2[k_c+\gamma]$, then it gets out of $\mathcal{Z}_2[k_c+\gamma+1]$. Thus, we have proved \eqref{nodegetout_deterministic}.

	Combining \eqref{notgetin_deterministic} and \eqref{nodegetout_deterministic}, we arrive at \eqref{shrinking_deterministic}. Therefore, for any $k \geq k_c + \overline{k}\cdot n_N $, it holds with positive probability that 
	\begin{align*}
		\big(\mathcal{Z}_1[k] \cup  \mathcal{Z}_2[k] \big)\cap \mathcal{N}= \emptyset,
	\end{align*}
	since the number of normal nodes is $n_N $ and each normal node will update at least once in $\overline{k}$ steps. We have the contradiction, and (C2) is proved.
	
	Lastly, we show (C3). Once the normal nodes have
	reached the common value $x^*$ at time $k$. Since $\left|\mathcal{N}_i^{l-} \cap \mathcal{A}\right| \leq f, \medspace \forall i \in \mathcal{N}$, according to step 2 of Algorithm~1, any node $j \in \mathcal{A}$ with value $x_j [k] \neq x^*$ is ignored by the normal nodes. Thus, $x_i [k + 1] = x^*, \forall i\in \mathcal{N}$. The proof is completed.
\end{proof}

\section*{D. Proof of Theorem \ref{theorem_byzan}}

\begin{proof}
	The necessity part is similar to the real-valued consensus in our previous work \cite{yuan2022asynchronous} and it is omitted.
	
	For the sufficiency part, we need to show the conditions (C1)–(C3) in Lemma \ref{synlemma} hold. Note that in the proof of Theorem \ref{asyndetermin}, conditions (C1) and (C3) still hold for Byzantine attacks. For condition (C2), the sets $\mathcal{Z}_1[k]$ and $\mathcal{Z}_2[k]$ are defined as the subsets of $\mathcal{N}$, i.e.,
	\begin{equation*}
		\begin{aligned}
			\mathcal{Z}_1[k]&=\{i\in \mathcal{N}: x_i[k]= \overline{z}^*\},\\
			\mathcal{Z}_2[k]&=\{i\in \mathcal{N}: x_i[k]= \underline{z}^*\},
		\end{aligned}
	\end{equation*}
	where $\overline{z}^*$ and $\underline{z}^*$ are, respectively, the final values (reached with probability 1) of the sequences of the maximum and minimum values of normal nodes across each $\tau+1$ time steps.
	Besides, node $i$ in one of these two sets uses at least one value from the normal nodes outside its corresponding set. Both situations are applied for normal nodes, and thus, the proof of Theorem \ref{asyndetermin} still holds even under the Byzantine model.
\end{proof}

\end{document}